\documentclass[runningheads, envcountsame]{llncs}
\usepackage[T1]{fontenc}
\usepackage{graphicx}
\usepackage[hidelinks]{hyperref}
\usepackage{color}

\usepackage[dvipsnames]{xcolor} 
\usepackage{amssymb}
\usepackage{amsmath}
\usepackage{booktabs}
\usepackage[frozencache=true,cachedir=minted-cache]{minted}
\usepackage{subcaption}

\usepackage{xspace}
\usepackage{cleveref}
\usepackage{fancyvrb}
\usepackage{mathtools}
\usepackage{tikz}
\usepackage[noend]{algpseudocode}
\usepackage{algorithm}
\usetikzlibrary{matrix,decorations.pathreplacing,tikzmark}
\usepackage{diagbox}
\usepackage{comment}
\usepackage{mathabx}
\usepackage{multido}
\usepackage[misc,geometry]{ifsym} 

\newcommand{\calP}{\mathcal{P}}

\usepackage{environ} 

\newcommand{\repeattheorem}[1]{%
  \begingroup
  \renewcommand{\thetheorem}{\ref{#1}}%
  \expandafter\expandafter\expandafter\theorem
  \csname reptheorem@#1\endcsname
  \endtheorem
  \endgroup
}

\NewEnviron{reptheorem}[1]{%
  \global\expandafter\xdef\csname reptheorem@#1\endcsname{%
    \unexpanded\expandafter{\BODY}%
  }%
  \expandafter\theorem\BODY\unskip\label{#1}\endtheorem
}

\newcommand{\repeatlemma}[1]{%
  \begingroup
  \renewcommand{\thelemma}{\ref{#1}}%
  \expandafter\expandafter\expandafter\lemma
  \csname replemma@#1\endcsname
  \endlemma
  \endgroup
}

\NewEnviron{replemma}[1]{%
  \global\expandafter\xdef\csname replemma@#1\endcsname{%
    \unexpanded\expandafter{\BODY}%
  }%
  \expandafter\lemma\BODY\unskip\label{#1}\endlemma
}

\newif\ifcomplver
\complverfalse

\begin{document}
\title{Stable Dinner Party Seating Arrangements}
%
%
\author{Damien Berriaud\orcidID{0009-0002-0450-0339} \and Andrei Constantinescu\textsuperscript{\thinspace\Letter\thinspace}\orcidID{0009-0005-1708-9376} \and
Roger Wattenhofer \orcidID{0000-0002-6339-3134}}
\authorrunning{D. Berriaud, A. Constantinescu and R. Wattenhofer}
%
\institute{ETH Zurich, Rämistrasse 101, 8092 Zurich, Switzerland \\
\email{\{dberriaud,aconstantine,wattenhofer\}@ethz.ch}}
\maketitle              
\begin{abstract}
A group of $n$ agents with numerical preferences for each other are to be assigned to the $n$ seats of a dining table. We study two natural topologies:~circular (cycle) tables and panel (path) tables. For a given seating arrangement, an agent's utility is the sum of their preference values towards their (at most two) direct neighbors. An arrangement is envy-free if no agent strictly prefers someone else's seat, and it is stable if no two agents strictly prefer each other's seats. Recently, it was shown that for both paths and cycles it is NP-hard to decide whether an envy-free arrangement exists, even for symmetric binary preferences. In contrast, we show that, if agents come from a bounded number of classes, the problem is solvable in polynomial time for arbitrarily-valued possibly asymmetric preferences, including outputting an arrangement if possible. We also give simpler proofs of the previous hardness results if preferences are allowed to be asymmetric. For stability, it is known that deciding the existence of stable arrangements is NP-hard for both topologies, but only if sufficiently-many numerical values are allowed. As it turns out, even constructing unstable instances can be challenging in certain cases, e.g., binary values. We propose a complete characterization of the existence of stable arrangements based on the number of distinct values in the preference matrix and the number of agent classes. We also ask the same question for non-negative values and give an almost-complete characterization, the most interesting outstanding case being that of paths with two-valued non-negative preferences, for which we experimentally find that stable arrangements always exist and prove it under the additional constraint that agents can only swap seats when sitting at most two positions away. Similarly to envy-freeness, we also give a polynomial-time algorithm for determining a stable arrangement assuming a bounded number of classes. We moreover consider the swap dynamics and exhibit instances where they do not converge, despite a stable arrangement existing.
\end{abstract}

\section{Introduction}
Your festive dinner table is ready, and the guests are arriving. As soon as your guests take their assigned seats, two of them are unhappy about their neighbors and
rather
want to switch seats. Alas, right after the switch, two other guests become upset, and then pandemonium ensues! Could you have prevented all the social awkwardness by seating your guests ``correctly'' from the get-go?

In this paper, we study the difficulty of finding a \emph{stable} seating arrangement; i.e.~one where no two guests would switch seats. We focus on two natural seating situations:~a round table (cycle), and an expert panel (path). In either case, we assume guests only care about having the best possible set of direct left and right neighbors.
In certain cases, not even a stable arrangement might make the cut, as a single guest envying the seat of another could potentially lead to trouble. Therefore, we are also interested in finding \emph{envy-free} arrangements.

Formally, $n$ guests have to be assigned bijectively to the $n$ seats of a dining table:~either a path or a cycle graph. Guests express their preferences for the other guests numerically, with higher numbers corresponding to a greater desire to sit next to the respective guest. The utility of a guest $g$ for a given seating arrangement is the sum of $g$'s preference values towards $g$'s neighbors. A guest $g$ envies another guest if $g$'s utility would strictly increase if they swapped places. Two guests want to swap places whenever they envy each other. Our goal is to compute a stable (no two guests want to swap) respectively envy-free (no guest envies another) seating arrangement.

Besides the table topology, we conduct our study in terms of two natural parameters.
The first parameter is the number of numerical values guests can choose from when expressing their preferences for other guests. For instance, an example of two-valued preferences would be when all preference values are either zero or one (i.e., \emph{binary}, also known as \emph{approval} preferences), in which case every guest has a list of ``favorite'' guests they want to sit next to, and is indifferent towards the others. In contrast, if the values used were $\pm 1$; i.e., every guest either likes or dislikes every other guest; then the preferences are still two-valued, but no longer binary. Increasing the number of allowed values allows for finer-grained preferences. It is also interesting to distinguish the special case of non-negative preferences; i.e., no guest can lose utility by gaining a neighbor. 

Our second parameter is the number of different guest classes. In particular, dinner party guests can often be put into certain categories, e.g., charmer, entertainer, diva, politico, introvert, outsider. Each class has its own preferences towards other classes, e.g., outsiders would prefer to sit next to a charmer, but not next to an introvert.

\makeatletter
\newcommand{\@ssymbol}[1]{%
\ifcase#1
\or{\textcolor{Red}{$\ast$}}%
\or{\textcolor{blue}{$\dag$}}%
\or{\textcolor{violet}{$\varstar$}}%
\or{\textcolor{black}{$\mathsection$}}%
\or{\textcolor{orange}{$\bullet$}}%
\else\@ctrerr
\fi}
\newcommand{\ssymbol}[1]{\@ssymbol{#1}}
\makeatother

\begin{table}[t]
    \begin{subtable}[h]{0.45\textwidth}
        \centering
        \begin{tabular}{c|cc}
                  & \multicolumn{2}{c}{No.~of Classes} \\
                  & Bounded     & Unbounded \\ \hline
            EF    & Poly\textsuperscript{\ssymbol{1}}        & NP-hard\textsuperscript{\ssymbol{2}} \\
            STA   & Poly\textsuperscript{\ssymbol{1}}        & NP-hard\textsuperscript{\ssymbol{3}} \\
        \end{tabular}
        \caption{Complexity results for cycles.  }
      \label{tab:summary_res_computational_cycle}
    \end{subtable}
    \hfill
    \begin{subtable}[h]{0.45\textwidth}
        \centering
        \begin{tabular}{c|cc}
                  & \multicolumn{2}{c}{No.~of Classes} \\
                  & Bounded     & Unbounded \\ \hline
            EF    & Poly\textsuperscript{\ssymbol{1}}        & NP-hard\textsuperscript{\ssymbol{2}} \\
            STA   & Poly\textsuperscript{\ssymbol{1}}        & NP-hard\textsuperscript{\ssymbol{4}} \\
        \end{tabular}
        \caption{Complexity results for paths.}
        \label{tab:summary_res_computational_path}
    \end{subtable}
    \caption{Summary of computational results for envy-free (EF) and stable (STA)
    arrangements.
    We distinguish between two cases, depending on whether the total number of guest classes is bounded by a constant or not. Hardness results are for the existential questions ``do such arrangements exist?''
    Polynomial-time algorithms recover an envy-free/stable arrangement whenever one exists. 
    \vspace{3pt}
    {\footnotesize
    \\ \ssymbol{1} are shown in Theorems \ref{thm:algo_k_class_path} and \ref{thm:algo_k_class_cycle}
    (working for arbitrary values).
    \\ \ssymbol{2} are shown in Theorems \ref{th_np_hard_cycle} and \ref{th_np_hard_path} (using binary values). They were also recently shown in \cite{hua} for symmetric binary preferences, although with more complex proofs.
    \\ \ssymbol{3} is shown in \cite{hua} using four non-negative values, the binary case is open.
    \\ \ssymbol{4} is shown in \cite{hua} using six values including negatives, the non-negative case is open.
    }
    }
    \label{tab:summary_res_computational}
\end{table}

\newcommand{\emptysuperscript}[1]{\phantom{\textsuperscript{\ssymbol{#1}}}}

\begin{table}[t]
    \begin{subtable}[h]{0.45\textwidth}
      \centering
      \begin{tabular}{c|cccc}
          \diagbox[linecolor=white,innerwidth=2.0cm]{Values}{Classes} & $\leq 2$ & $ 3$ & 4 & $ \geq 5$ \\ \hline
          $ \leq 2$                       & S\emptysuperscript{1}        & \textbf{S}\textsuperscript{\ssymbol{3}}    & \textbf{U}\textsuperscript{\ssymbol{4}} & U \\
          3                               & S\emptysuperscript{1}        & \textbf{U}\textsuperscript{\ssymbol{2}}    & U\emptysuperscript{3} & U \\
          $\geq 4$                        & \textbf{S}\textsuperscript{\ssymbol{1}}        & U\emptysuperscript{2}    & U\emptysuperscript{3} & U \\
      \end{tabular} 
      \caption{Characterization results for cycles.}
      \label{tab:summary_res_cycle}
    \end{subtable}
    \hfill
    \begin{subtable}[h]{0.45\textwidth}
      \centering
      \begin{tabular}{c|ccc}
          \diagbox[linecolor=white,innerwidth=2.0cm]{Values}{Classes} & $\leq 2$ & $ 3$ & $\geq 4$  \\ \hline
          $ \leq 2$                       & S\emptysuperscript{1}        & \textbf{U}\textsuperscript{\ssymbol{5}}    & U\textsuperscript{\ssymbol{5}} \\
          3                               & S\emptysuperscript{1}        & \textbf{U}\textsuperscript{\ssymbol{2}}    & U\emptysuperscript{5} \\
          $\geq 4$                        & \textbf{S}\textsuperscript{\ssymbol{1}}        & U\emptysuperscript{5}    & U\emptysuperscript{5} \\
      \end{tabular} 
      \caption{Characterization results for paths.}
      \label{tab:summary_res_path}
    \end{subtable}
    \caption{Summary of results characterizing the existence of unstable instances for different combinations of constraints on the number of preference values and classes of guests. Our stability (S) results mean that all instances satisfying the constraints admit a stable arrangement, and hold for arbitrary preference values. Our constructions with no stable arrangements (U) only use small non-negative values (except \ssymbol{5}, discussed below), often $0, 1, 2, \ldots,$ and work for any large enough number of guests. 
    \vspace{3pt}
    {\footnotesize
    \\ \ssymbol{1} are shown in 
 Theorems \ref{thm:2class_cycle} and \ref{thm:2_class_path_stable} (working for arbitrary values). \\ 
    \ssymbol{2} are shown in Theorems \ref{thm:unst_ternary_cycle} and \ref{thm:unst_tern_path} (using non-negative values). \\
    \ssymbol{3} is shown in Theorem \ref{thm:3_class_2_valued_cycle_stable} (working for arbitrary values). \\
    \ssymbol{4} is shown in Theorem \ref{thm:bin_unst_cycle} (using binary values). \\
    \ssymbol{5} is shown in Theorem \ref{thm:3_class_2_val_unst_path} (using also negative values). The case of non-negative values is open, and we believe that the answer is \textbf{S}. We have exhaustively established this for 4-class instances with at most ten guests per class and 5-class instances with at most four guests per class, as well as all 7-guest instances. Moreover, it holds irrespective of the number of classes assuming that guests are only willing to swap places with other guests that they are separated from by at most one seat (Theorem \ref{thm:close-swaps-two-valued-is-stable}).  
    }
    }
    \label{tab:summary_res}
\end{table}


\noindent \textbf{Our Contribution}. We study the existence and computational complexity of finding stable/envy-free arrangements on paths and cycles. Some of our results can be surprising. For instance, six people with binary preferences can always be stably seated at a round table, while for five (or seven) guests some preferences are inherently unstable, so we better invite (or uninvite) another guest. However, even for six people with binary preferences, for which a stable arrangement always exists, the swap dynamics might still never converge to one.

Our computational results are exhibited in Table \ref{tab:summary_res_computational}. In summary, if the number of guest classes is bounded by a constant (which can be arbitrary), then both stable and envy-free arrangements can be computed in polynomial time whenever they exist with no assumption on the preference values, while dropping this assumption makes the two problems difficult even for very constrained preference values. For envy-freeness, this already happens for binary preferences, arguably the most prominent case. For stability, on the other hand, this requires more contrived values (four non-negative values for cycles and six values including negatives for paths), so it would be interesting to have a finer-grained understanding of stability in terms of the values allowed when expressing one's preferences. For instance, is it still hard to find stable arrangements for binary preferences? As we show, it turns out to be surprisingly difficult to even construct unstable binary preferences, setting aside the computational considerations. To this end, we conduct a fine-grained study aiming to answer for which combinations of our two parameters, i.e., number of preference values and guest classes, do stable arrangements always exist and for which combinations this is not the case. Our results are exhibited in Table \ref{tab:summary_res}. Notably, for cycles we give a full characterization, either showing guaranteed stability with arbitrary values or giving unstable instances with (simple) non-negative values. Similarly, for paths, we close all cases, with the notable exception of two-valued preferences with three or more classes being permitted, where all the unstable instances we found require negative values. We conjecture that for non-negative values stability on paths is guaranteed in the two-valued case. We support this conjecture with experimental evidence, as well as a partial result under the additional assumption of guests only being willing to swap seats if they are separated by at most one seat (see Table \ref{tab:summary_res} for more details).

\noindent \textbf{Appendix}. In the appendix, we supply the proofs omitted from the main text, as well as supporting material.
Moreover, we show that stability is a highly fra\-gile notion, being non-monotonic with respect to adding/removing guests. We also give evidence that knowledge about stability on paths is unlikely to transfer to computing cycle-stable arrangements. Finally, we use probabilistic tools to study the expected number of stable arrangements of Erd\H{o}s-Rényi binary preferences. \\


\section{Related Work}

The algorithmic study of stability in collective decision-making has its roots in the seminal paper of Gale and Shapley \cite{gale_college_1962}, introducing the now well-known \emph{Stable Marriage} and \emph{Stable Roommates} problems. Classically, the former is presented as follows:~an equal number of men and women want to form couples such that no man and woman from different couples strictly prefer each other over their current partners, in which case the matching is called \emph{stable}. The authors give the celebrated Gale-Shapley deferred acceptance algorithm showing that a stable matching always exists and can be computed in linear time. Irving \cite{irving_stable_marriage_indifference} later extended the algorithm to also handle preferences with ties; i.e., a man (woman) being indifferent between two women (men). The \emph{Stable Roommates} problem is the non-bipartite analog of Stable Marriage:~an even number of students want to allocate themselves into identical two-person rooms in a dormitory. A matching is stable if no two students allocated to different rooms prefer each other over their current roommates. In this setting, stable matchings might no longer exist, but a polynomial-time algorithm for computing one if any exist is known \cite{irving_efficient_1985}. However, when ties are allowed, the problem becomes NP-hard \cite{ronn_np-complete_1990}.

The seating arrangement problem that we study is, in fact, well-connected with Stable Roommates. Instead of one table with $n$ seats, the latter considers $n / 2$ tables with two seats each. However, there is another more subtle difference:~in Stable Roommates, two people unhappy with their current roommates can choose to move into any free room. This is not possible if there are exactly $n / 2$ rooms. Instead, the stability notion that we study corresponds to the distinct notion of \emph{exchange-stability} in the Stable Roommates model, where unhappy students can agree to exchange roommates. Surprisingly, under exchange-stability, finding a stable roommate allocation becomes NP-hard even without ties \cite{cechlarova_complexity_2002}.

One can also see our problem through the lens of coalition formation. In particular, \emph{hedonic games} \cite{handbook_hedonic} consider the formation of coalitions under the assumption that individuals only care about members in their own coalition. Then, fixing the sizes of the coalitions allows one to generalize from tables of size two and study stability more generally. Bilò et al.~\cite{bilo_hedonic_2022} successfully employ this approach to show a number of attractive computational results concerning exchange-stability. The main drawback of this approach, however, is that it assumes that any two people sitting at the same table can communicate, which is not the case for larger tables. Our approach takes the topology of the dining table into account.

Some previous works have also considered the topology of the dining table. Perhaps closest to our paper is the model of Bodlaender et al.~\cite{bodlaender_hedonic_2020},
in which $n$ individuals are to be assigned to the $n$ vertices of an undirected seating graph. The authors prove a number of computational results regarding both envy-freeness and exchange-stability, among other notions. However, we found some of the table topologies considered to be rather unnatural, especially in hardness proofs (e.g., trees or unions of cliques and independent sets). Bullinger and Suksompong \cite{bullinger_topological_2022} also conduct an algorithmic study of a similar problem, but with a few key differences: (i) individuals are seated in the nodes of a graph, but there may be more seats than people; (ii) for the stability notion, they principally consider \emph{jump-stability}, where unhappy people can choose to move to a free seat; (iii) 
individuals now contribute to everyone's utility, although their contribution decreases with distance.

Last but not least, studying stability in the context of \emph{Schelling games} has recently been a popular area of research \cite{chauhan_schelling_2018,elkind_schelling_2021,bilo_nash_2018,kreisel_equilibria_2022,bilo_topological_2020}. In Schelling games, individuals belong to a fixed number of classes. However, unlike in our model, agents from one class only care about sitting next to others of their own class. This additional assumption often allows for stronger results; e.g., in \cite{bilo_topological_2020} the authors prove the existence of exchange-stable arrangements on regular and almost regular topological graphs such as cycles and paths, and show that the swap dynamics are guaranteed to converge in polynomial time on such topologies.

Overall, it seems that exchange stability has been studied in the frameworks of both hedonic and Schelling games. However, both approaches present some shortcomings:~on the one hand, Schelling games inherently consider a topology on which agents evolve, but, being historically motivated by the study of segregation (e.g., ethnic, racial), they usually restrict themselves to very simple preferences. On the other hand, works on hedonic games are accustomed to considering diverse preferences. However, while multiple works have introduced topological considerations, their analysis is usually constrained to graphs that can be interpreted as non-overlapping coalitions, e.g., with multiple fully connected components. One notable exception is the very recent work of Ceylan, Chen and Roy \cite{hua}, appearing in IJCAI'23, also building on the model of Bodlaender et al.~\cite{bodlaender_hedonic_2020}.
In comparison to us, they also prove the NP-hardness of deciding the existence of envy-free arrangements for binary preferences for both topologies, in their case for symmetric preferences, but at the expense of more complex proofs. They moreover show that hardness holds for stability, although the presented proofs requires four non-negative values for cycles and six values including negatives for paths. In our work, we aim to understand stability under more natural preference values, such as approval preferences.

\section{Preliminaries}\label{sec:prelims}
We write $[n] = \{1,\ldots,n\}.$ Given an undirected graph $G = (V(G), E(G)),$ we write $N_G(v)$ for the set of neighbors of vertex $v\in V(G).$ When clear from context, oftentimes we will simply write $V, E$ and $N(v)$ respectively.

The model we describe next is similar to the one in \cite{bodlaender_hedonic_2020}.
A group of $n$ agents (guests) $\mathcal{A}$
has to be seated at a dining table represented by an undirected graph $G=(V, E),$ where vertices correspond to seats. We will be interested in the cases of $G$ being a cycle or a path. We assume that $|V| = n$ and that no two agents can be seated in the same place, from which it also follows that all the seats have to be occupied. Agents have numerical preferences over each other, corresponding to how much utility they gain from being seated next to other agents. In particular, each agent $i \in \mathcal{A}$ has a \emph{preference} over the other agents expressed as a function $p_{i}:~\mathcal{A} \setminus \{i\} \to \mathbb{R},$ where $p_i(j)$ denotes the utility gained by agent $i$ when sitting next to $j.$ Note that we do not assume symmetry; i.e., it might be that $p_i(j) \neq p_j(i).$ We denote by $\mathcal{P} = (p_i)_{i \in \mathcal{A}}$ the collection of agent preferences, or \emph{preference profile}, of the agents. A number of different interpretations can be associated to $\mathcal{P}.$ In particular, we will usually see $\mathcal{P}$ as a matrix $\mathcal{P} = (p_{ij})_{i, j \in \mathcal{A}},$ where $p_{ij} = p_{i}(j).$ Note that the diagonal entries are not defined, but, for convenience, we will oftentimes use the convention that $p_{ii} = 0.$
Using the matrix notation, we say that the preferences in $\mathcal{P}$ are
\emph{binary} when $\mathcal{P} \in \{0,1\}^{n \times n}$
and \emph{k-valued} if there exists $\Gamma \subseteq \mathbb{R}$, $|\Gamma| = k$, such that $\mathcal{P} \in \Gamma^{n \times n}$ (disregarding diagonal entries, since they are undefined).
Note that binary preferences are two-valued, but two-valued preferences are not necessarily binary. Moreover, we will often represent binary preferences as a directed graph, where a directed edge between two agents signifies that the first agent approves of the second.
Finally, when $p_{ij} \geq 0$ for any two agents $i, j \in \mathcal{A}$ we say that the preferences are \emph{non-negative}.

We define a \emph{class of agents} to be a subset of indistinguishable agents $\mathcal{C} \subseteq \mathcal{A}.$ More formally, all agents in $\mathcal{C}$ share a common preference function $p_\mathcal{C}:~\mathcal{A} \to \mathbb{R}$ and no agent in $\mathcal{A}$ discriminates between two agents in $\mathcal{C}.$ Note that this implies that the lines and columns of the preference matrix corresponding to agents in $\mathcal{C}$ are identical, if we adopt the convention that diagonal entries inside a class are all equal but not necessarily zero. We say that preference profile $\mathcal{P}$ has \emph{$k$ classes}, or is a \emph{$k$-class profile}, if $\mathcal{A}$ can be partitioned into $k$ classes $\mathcal{C}_1 \cup \ldots \cup \mathcal{C}_k = \mathcal{A}.$

We define an \emph{arrangement} of the agents on $G$ to be a bijection $\pi:~\mathcal{A} \to V(G)$, i.e., an assignment of each agent to a unique vertex of the seating graph (and vice-versa).
For a given arrangement $\pi,$ we define for each agent $i \in \mathcal{A}$ their \emph{utility} $U_{i}(\pi) = \sum_{v \in N_G(\pi(i))} p_i(\pi^{-1}(v))$ to be the sum of agent $i$'s preferences towards their graph neighbors in the arrangement. We say that agent $i$ \emph{envies} agent $j$ whenever $U_{i}(\pi) < U_{i}(\pi')$, where $\pi'$ is $\pi$ with $\pi(i)$ and $\pi(j)$ swapped.
We further say that $(i, j)$ is a \emph{blocking pair} if both $i$ envies $j$ and $j$ envies $i$; i.e., they would both strictly increase their utility by exchanging seats. An arrangement is \emph{envy-free} if no agent envies another, and it is \emph{stable} if it induces no blocking pairs. Note that envy-freeness implies stability, but the converse is not necessarily true.
By extension, we call preference profile $\mathcal{P}$ stable (respectively envy-free) on $G$ if there exists a stable (respectively envy-free) arrangement $\pi$ on $G.$ Profile $\mathcal{P}$ is unstable if it is not stable.

A few preliminary observations follow. Note that, for \emph{symmetric} preferences; i.e., $p_{ij} = p_{ji}$ for any two agents $i, j \in \mathcal{A};$ a stable arrangement always exists, in fact on any graph $G$, not just cycles and paths. This is because swaps in that case strictly increase the total sum of agent utilities, and hence the swap dynamics will converge to a stable arrangement. Hence, the interesting case is the asymmetric one. Observe that envy-free arrangements need not exist for symmetric preferences (deciding existence is NP-hard for both paths and cycles \cite{hua}). Moreover, note that, for cycles, and in fact any regular graph $G$, agents have the same number of neighbors, so transforming the preferences of an agent $i$ by adding/subtracting/multiplying by a positive constant the values $p_{ij}$ does not inherently change the preferences. This implies that for cycles the two-valued case coincides with the binary case and the non-negative case coincides with the unconstrained case. It also shows that for cycles the definition of $k$-valued preferences:~$\mathcal{P} \in \Gamma^{n \times n}$ where $|\Gamma| = k$; can be restated equivalently to require that every row of $\mathcal{P}$ consists of at most $k$ different values.

\section{Envy-Freeness}\label{sec:envy-free}
It is relatively easy to construct preferences with no envy-free arrangements:~for paths, even if all agents like each other, the agents sitting at the endpoints will envy the others; for cycles, add an agent despised by everyone, agents sitting next to them will envy their peers. We now show that, furthermore, it is NP-hard to decide whether envy-free arrangements exist, for both paths and cycles, even under binary preferences. This has also been recently shown in \cite{hua} in a stronger form, using only symmetric binary preferences, but we found the construction rather involved and the correctness argument based on careful counting delicate.

\begin{theorem} \label{th_np_hard_cycle} For binary preferences, deciding whether an envy-free arrangement on a cycle exists is NP-hard.
\end{theorem}

\begin{proof} We proceed by reduction from Hamiltonian Cycle on directed graphs. Let $G = (V, E)$ be a directed graph\footnote{In this proof $G$ is not the seating graph but rather an arbitrary graph.} such that, without loss of generality, $V = [n].$ If $G$ has any vertices with no outgoing edges, then map the input instance to a canonical no-instance, unless $n = 1,$ in which case we map it to a canonical yes-instance. Hence, from now on, assume that all vertices have outgoing edges. For each vertex $v \in V$ introduce three agents $x_v, y_v, z_v$ such that agent $x_v$ only likes $y_v$ and dislikes\footnote{Technically, is indifferent to everyone else, but found this formulation reads better.} everyone else, agent $y_v$ only likes $z_v$ and dislikes everyone else, and agent $z_v$ likes agent $x_u$ for all $u \in V$ such that $(v, u) \in E,$ and dislikes everyone else. We claim that the so-constructed preference profile $\mathcal{P}$ has an envy-free arrangement on a cycle precisely when $G$ has a Hamiltonian cycle. To show this, first assume without loss of generality that $1 \rightarrow 2 \rightarrow \ldots \rightarrow n \rightarrow 1$ is a Hamiltonian cycle in $G.$ Then, arranging agents around the cycle in the order $x_1, y_1, z_1, x_2, y_2, z_2, \ldots, x_n, y_n, z_n$ is an envy-free arrangement. To see why, notice that in this arrangement all agents get utility 1, so envy could only potentially stem from an agent being able to swap places with another agent to get utility 2. To prove this is not possible, first notice that agents $(x_i)_{i \in [n]}$ and $(y_i)_{i \in [n]}$ each only like one other agent, so they can never get a utility of more than 1 in any arrangement. Moreover, no agent $z_i$ can get to a utility of 2 by a single swap because any two agents they like are seated at least three positions from each other on the cycle. Conversely, assume that an envy-free arrangement $\pi$ exists. First, if $x_i$ is not sitting next to $y_i$ in $\pi,$ then $x_i$ could improve by swapping to a place next to $y_i$ (also similarly for $y_i$ and $z_i$). Therefore, in arrangement $\pi$ agent $x_i$ is seated next to $y_i$ and $y_i$ is seated next to $z_i.$ Moreover, consider the other neighbor of $z_i$ in $\pi.$ Since $z_i$ does not like $y_i,$ it follows that if $z_i$ also does not like their other neighbor, then $z_i$ could strictly improve their utility by swapping next to some agent they like, which is always possible because all vertices in $G$ have outgoing edges. Therefore, the other neighbor of $z_i$ has to be some agent that they like, hence being of the form $x_j,$ where $j \neq i$. Note that, by construction, $(i, j) \in E.$ Putting together what we know, we get that under $\pi$ the agents are arranged around in the cycle in some order $x_{\sigma_1}, y_{\sigma_1}, z_{\sigma_1}, x_{\sigma_2}, y_{\sigma_2}, z_{\sigma_2}, \ldots, x_{\sigma_n}, y_{\sigma_n}, z_{\sigma_n},$ where $\sigma$ is a permutation of the $n$ agents such that $(\sigma_i, \sigma_{i + 1}) \in E$ holds for all $i \in [n].$\footnote{Assuming that addition is performed with wrap-around using $n + 1 \equiv 1 \pmod n$.} Therefore, a Hamiltonian cycle $\sigma_1 \rightarrow \sigma_2 \rightarrow \ldots \rightarrow \sigma_n \rightarrow \sigma_1$ exists in $G.$  
\end{proof}

A similar proof can be used to show hardness for the case of paths. We outline the changes required in the appendix.

\begin{reptheorem}{th_np_hard_path} For binary preferences, deciding whether an envy-free arrangement on a path exists is NP-hard.
\end{reptheorem}

Proving matching hardness results for stability under binary preferences would be highly desirable but for the time being is left open. We in fact conjecture that all instances with binary preferences are stable on a path (see Section \ref{sec:binary_pref}). Moreover, one might ask how does the number of agent classes affect the computational complexity of our problems. In the next section, we address this question, showing that limiting the number of agent classes renders the problems that we consider polynomial-time solvable, even for arbitrary preference values.

\section{Polynomial Solvability for $k$-Class Preferences}\label{sec:k_classes_algo}
In this section, we show that deciding whether envy-free and stable arrangements exist for a given preference profile can be achieved in polynomial time, for both paths and cycles, assuming that the number of agent classes is bounded by a number $k$. Note that preferences in this case are not constrained to being binary, and can in fact be arbitrary. By extension, our algorithms can also be used to construct such arrangements whenever they exist.

We begin with the case of paths. For simplicity, we assume that $n \geq 3$, as for $n \leq 2$ any arrangement is both stable and envy-free. Assume that the agent classes are identified by the numbers $1, \ldots, k$ and that $n_1, n_2, \ldots, n_k$ are the number of agents of each class in our preference profile, where $n_1 + \ldots + n_k = n$. For ease of writing, we will see arrangements as sequences $s = (s_i)_{i \in [n]}$, where $s_i \in [k]$ and for any agent class $j \in [k]$ the number of values $j$ in $s$ is $n_j.$ Moreover, for brevity, we lift agent preferences to class preferences, in order to give meaning to statements such as ``class $a$ likes class $b.$'' To simplify the treatment of agents sitting at the ends of the path, we introduce two agents of a dummy class 0 with preference values 0 from and towards the other agents. We require the dummy agents to sit at the two ends of the path; i.e., $s_0 = s_{n + 1} = 0.$ In order to use a common framework for stability and envy-freeness, we define the concept of \emph{compatible triples} of agent classes, as follows. First, for envy-freeness, let $a, b, c, d, e, f$ be agent classes, then we say that triples $(a, b, c)$ and $(d, e, f)$ are \emph{long-range compatible} if $p_b(a) + p_b(c) \geq p_b(d) + p_b(f)$ and $p_e(d) + p_e(f) \geq p_e(a) + p_e(c)$; intuitively, neither $b$ wants to swap with $e$, nor vice-versa. Furthermore, for $a, b, c, d$ agent classes, we say that triples $(a, b, c)$ and $(b, c, d)$ are \emph{short-range compatible} if $p_b(a) \geq p_b(d)$ and $p_c(d) \geq p_c(a)$; intuitively, if $a, b, c, d$ are consecutive in the arrangement, then neither $b$ wants to swap with $c$, nor vice-versa. For stability, we keep the same definitions but use ``or'' instead of ``and.'' Note that long-range and short-range compatibility do not imply each other.
We call an arrangement $s$ \emph{compatible} if for all $1 \leq i < j \leq n$ the triplets $(s_{i - 1}, s_i, s_{i + 1})$ and $(s_{j - 1}, s_j, s_{j + 1})$ are long-range compatible when $j - i > 1$ and short-range compatible when $j - i = 1$. Note that arrangement $s$ is envy-free (resp.~stable) if and only if it is compatible. In the following, we explain how to decide the existence of a compatible arrangement.

\begin{lemma} Deciding whether compatible arrangements exist can be achieved in polynomial time.
\end{lemma}
\begin{proof}
We first present a nondeterministic algorithm (i.e.,~with guessing) that solves the problem in polynomial time. The algorithm builds a compatible arrangement $s$ one element at a time. Initially, the algorithm sets $s_0 \gets 0$ and guesses the values of $s_1$ and $s_2$. Then, at step $i,$ for $3 \leq i \leq n + 1,$ the algorithm will guess $s_i$ (except for $i = n + 1$, where we enforce that $s_i \gets 0$) and check whether $(s_{i - 3}, s_{i - 2}, s_{i - 1})$ is short-range conflicting with $(s_{i - 2}, s_{i - 1}, s_i),$ rejecting if so. Moreover, the algorithm will check whether $(s_{i - 2}, s_{i - 1}, s_i)$ is long-range conflicting with any $(s_{j - 2}, s_{j - 1}, s_j)$ for $2 \leq j \leq i - 2,$ again rejecting if so. At the end, the algorithm checks whether for each $i \in [k]$ value $i$ occurs in $s$ exactly $n_i$ times, accepting if so, and rejecting otherwise.

Alone, this algorithm only shows containment in NP, which is not a very attractive result. Next, we show how the same algorithm can be implemented with only a constant number of variables, explaining afterward why this implies our result. First, to simulate the check at the end of the algorithm without requiring knowledge of the whole of $s,$ it is enough to maintain throughout the execution counts $(x_j)_{j \in [k]}$ such that at step $i$ in the algorithm $x_j$ gives the number of positions $1 \leq \ell \leq i$ such that $s_\ell = j.$ To simulate the short-range compatibility check, it is enough that at step $i$ we have knowledge of $s_{i - 3}, \ldots, s_i.$ Finally, for the long-range compatibility check, a more insightful idea is required. In particular, we make the algorithm maintain throughout the execution counts $m_{a, b, c}$ for each triple $(a, b, c)$ of agent classes, such that at step $i$ value $m_{a, b, c}$ gives the number of positions $2 \leq \ell \leq i$ such that $(s_{\ell - 2}, s_{\ell - 1}, s_\ell) = (a, b, c).$ Using this information, to check at step $i$ whether $(s_{i - 2}, s_{i - 1}, s_i)$ long range conflicts with any $(s_{j - 2}, s_{j - 1}, s_j)$ for $2 \leq j \leq i - 2,$ it is enough to temporarily decrease by one the values $m_{s_{i - 3}, s_{i - 2}, s_{i - 1}}$ and $m_{s_{i - 2}, s_{i - 1}, s_i}$ and then check whether there exists a triple $(a, b, c)$ of agent classes such that $m_{a, b, c} > 0$ and $(a, b, c)$ long-range conflicts with $(s_{i - 2}, s_{i - 1}, s_i).$ In total, at step $i$, the algorithm only needs to know the values $s_{i - 3}, \ldots, s_i$, as well as $(x_j)_{j \in [k]}$ and the counts $m_{a, b, c}$ for all triples $(a, b, c)$ of agent classes. Since $k + 1$ bounds the total number of agent classes, this is only a constant number of variables. As each variable can be represented with $O(\log n)$ bits, it follows that our nondeterministic algorithm uses only logarithmic space, implying containment in the corresponding complexity class NL. It is well known that NL $\subseteq$ P, from which our conclusion follows. For readers less familiar with this result, we give a short overview of how our algorithm can be converted into a deterministic polynomial-time algorithm, as follows. Since our NL algorithm uses only logarithmic space, it follows that the space of algorithm states that can be reached depending on the nondeterministic choices is of at most polynomial size, since $2^{O(\log n)}$ is polynomial. Therefore, building a graph with vertices being states and oriented edges corresponding to transitions between states, the problem reduces to deciding whether an accepting state can be reached from the initial state, which can be done with any efficient graph search algorithm.
\end{proof}

\begin{theorem} \label{thm:algo_k_class_path} Fix $k \geq 1.$ Then, for $k$-class preferences, there are polynomial-time algorithms computing an envy-free/stable arrangement on a path or reporting the nonexistence thereof.
\end{theorem}

For the case of cycles, a similar approach can be used, although with rather tedious, yet minor tweaks, presented in the appendix.

\begin{reptheorem}{thm:algo_k_class_cycle} Fix $k \geq 1.$ Then, for $k$-class preferences, there are polynomial-time algorithms computing an envy-free/stable arrangement on a cycle or reporting the nonexistence thereof.
\end{reptheorem}

\section{A Fine-Grained Analysis of Stability}\label{sec:stab}

Previously, we showed that deciding whether envy-free arrangements exist is NP-hard for binary preferences on both topologies. For stability, in \cite{hua} it is shown that  similar hardness results hold, but this time more than two values are needed in the proofs, namely four values for cycles and six for paths, the latter also requiring negative numbers. It is unclear whether hardness is retained without assuming this level of preference granularity, and a first step towards understanding the difficulty of the problem constrained to fewer/simpler allowed values is being able to construct instances where no stable arrangements exist; after all, a problem where the answer is always ``yes'' cannot be NP-hard.


In this section, we conduct a fine-grained analysis of the conditions allowing for unstable instances. In particular, for both topologies we consider how different constraints on the number of agent classes as well as the number of different values allowed in the preferences influence the existence of unstable preferences
The non-negativity of the values needed is also taken into account. Table \ref{tab:summary_res} summarizes our results.

\subsection{Two-Class Preferences} \label{sec:2_classes}

As a warm-up, note that when all agents come from a single class, any arrangement on any given seating graph is stable. In the following, we extend this result to two classes of agents for cycles and paths. We begin with cycles:

\begin{theorem}\label{thm:2class_cycle} Two-class preferences always induce a stable arrangement on a cycle.
\end{theorem}

\begin{proof}
    Suppose there are two classes of agents, say \textcolor{blue}{Blues} and \textcolor{Red}{Reds}.
    Without loss of generality, preferences can be assumed to be binary, since for cycles one can normalize the preference values as described towards the end of the preliminaries section.  First, note that any blocking pair must consist of one \textcolor{blue}{Blue} and one \textcolor{Red}{Red}. Moreover, note that any arrangement is stable whenever one of the classes likes the two classes equally. Now, suppose this is not the case, meaning that each class has a preferred class to sit next to. There are only two cases to consider:
    
    If one class, say \textcolor{blue}{Blue}, prefers its own class, then sit all \textcolor{blue}{Blues} together and give the remaining seats to \textcolor{Red}{Reds}:~all \textcolor{blue}{Blues} but two, say $\textcolor{blue}{B_1}$ and $\textcolor{blue}{B_2}$, get maximum utility, and neither $\textcolor{blue}{B_1}$ nor $\textcolor{blue}{B_2}$ can improve since no \textcolor{Red}{Red} has more than one \textcolor{blue}{Blue} neighbor. Hence, no \textcolor{blue}{Blue} is part of a blocking pair, so the arrangement is stable.
    
    If both classes prefer the opposite class, we may assume there are at least as many \textcolor{Red}{Reds} as \textcolor{blue}{Blues}. Then we alternate between \textcolor{Red}{Reds} and \textcolor{blue}{Blues} for as long as there are \textcolor{blue}{Blues} without a seat, then seat all the remaining \textcolor{Red}{Reds} next to each other. Every \textcolor{blue}{Blue} has maximum utility, and hence cannot be part of a blocking pair, so the arrangement is stable. 
\end{proof}

The following extends the result to the case of paths. The proof is largely similar, but the case analysis becomes more involved, because preferences can no longer be assumed to be binary, so we present it in the appendix.

\begin{reptheorem}{thm:2_class_path_stable} Two-class preferences always induce a stable arrangement on a path.
\end{reptheorem}

Note that a path of $n$ is equivalent to a cycle of $n+1$ where an agent with null preferences is added. This explains why the case of paths is harder to study than that of cycles, as it corresponds to having one more class of agents and potentially one more value (zero). 

\subsection{Three-Class Three-Valued Preferences}\label{sec:3-val_3-class}
We now consider the case of three-valued preferences with three agent classes, exhibiting unstable non-negative preferences both for paths and for cycles. We begin with the case of cycles.

\begin{theorem}\label{thm:unst_ternary_cycle}
For $n\ge 4$, there exist three-class three-valued non-negative preferences such that all arrangements on a cycle are unstable.
\end{theorem}

\begin{proof}
Consider three classes of agents:~Alice, Bob, and $n - 2$ of Bob's friends. The story goes as follows:~Alice and Bob broke up. Alice does not want to hear about Bob and would hence prefer to sit next to any of his friends rather than Bob. On the other hand, Bob wants to win her back, so he would above all want to sit next to Alice. Finally, Bob's friends prefer first Bob, then the other friends, and finally Alice. 
Used preference values can be arbitrary, so to get the required conclusion, we make sure that they are non-negative.
To show that these preferences are unstable on a cycle, there are two cases:~Alice and Bob can either sit next to each other, or separately.\footnote{The swap dynamics here will exhibit so-called run-and-chase behaviour \cite{run_and_chase}, which is common to many classes of hedonic games.}

In the first case, Alice and her second neighbor, who is one of Bob's friends, would exchange seats. After the switch, Alice is better as she no longer sits next to Bob, and the friend is better because he sits next to Bob.

In the second case, Bob and one of Alice's neighbors would exchange seats. Bob is better because he now sits next to Alice. To see that the neighbor, who is one of Bob's friends, is also better, distinguish two sub-cases:~if the friend sits right between Alice and Bob, then he is better because he now no longer sits next to Alice, while if this is not the case, he is better because before he was sitting next to Alice and a friend, while now he is sitting next to two friends.
\end{proof}

It is possible to use the same construction for paths by allowing negative preference values, but otherwise, the proof of Theorem \ref{thm:unst_ternary_cycle} does not directly transfer to paths; e.g., for $n=5$, path arrangement (F, F, B, A, F) is stable. Negative preferences turn out to not be necessary for $n$ large enough.
The main trick here is to use three copies of Bob to ensure that at least one of them does not sit at either end of the path. Formally, we have the following, proven in the appendix:

\begin{reptheorem}{thm:unst_tern_path}
For $n\geq 12$, there exist three-class three-valued non-negative preferences such that all arrangements on a path are unstable. 
\end{reptheorem}

\subsection{Two-Valued Preferences} \label{sec:binary_pref}

It remains to study what happens for two-valued preferences with three or more classes of agents. For cycles, we show that three classes always yield a stable arrangement, while four classes allow for a counterexample with binary values. For paths, we exhibit a counterexample with already three classes, but using some negative values. For non-negative values, we conjecture that a stable arrangement always exists, but have not been able to prove it. Instead, we gather both experimental and theoretical evidence to support it.

\begin{table}[t]
\centering
\addtolength{\tabcolsep}{7pt} 
\begin{tabular}
{c|ccccc}
$n$  & 3 & 4 & 5 & 6 & 7 \\ \hline
Cycle & 0 & 0 & 1 & 0 & 2\\ \hline
Path  & 0 & 0 & 0 & 0 & 0 
\end{tabular}
\addtolength{\tabcolsep}{-7pt}
\caption{The number of non-isomorphic families of unstable non-negative two-valued preferences. For cycles, this coincides with the case of binary values, and also with that of general values.}
\label{tab:summary_bin_results}
\end{table}
\begin{figure}[t]    
    \centering
    \begin{subfigure}{0.32\linewidth}
        \centering
        \includegraphics[width=3.5cm]{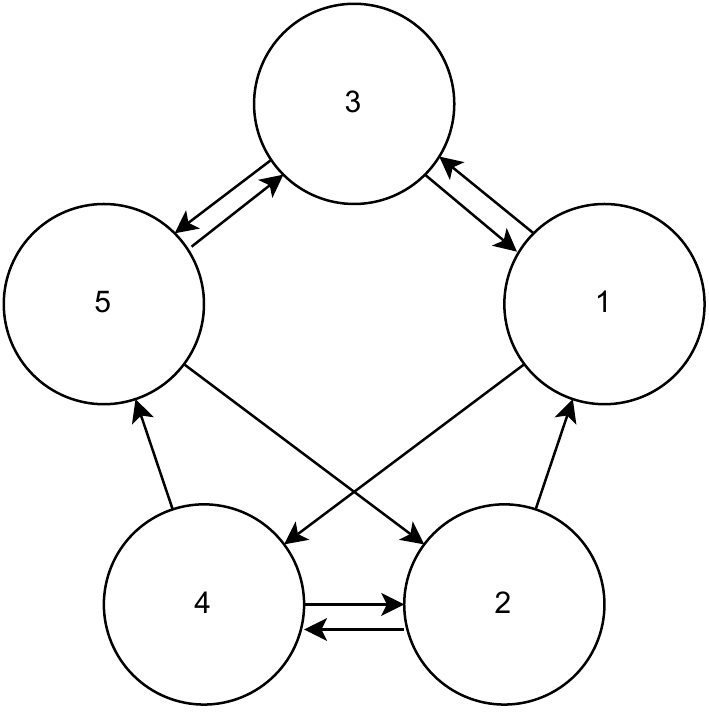}
        \caption{$\mathcal{P}_5$ as a directed graph.}
        \label{fig:5_unst}    
    \end{subfigure}
    \hfill
    \begin{subfigure}{0.32\linewidth}
        \centering
        \includegraphics[width=3.5cm]{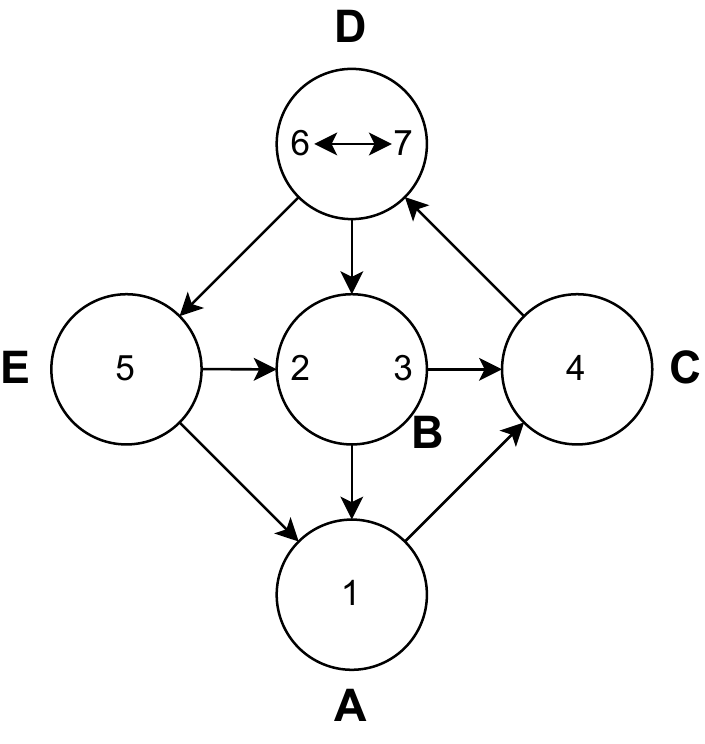}
        \caption{$\mathcal{P}_7^{(1)}$ as a directed graph.}
        \label{fig:7_unstable_1}
    \end{subfigure}
    \hfill
    \begin{subfigure}{0.32\linewidth}
        \centering
        \includegraphics[width=3.5cm]{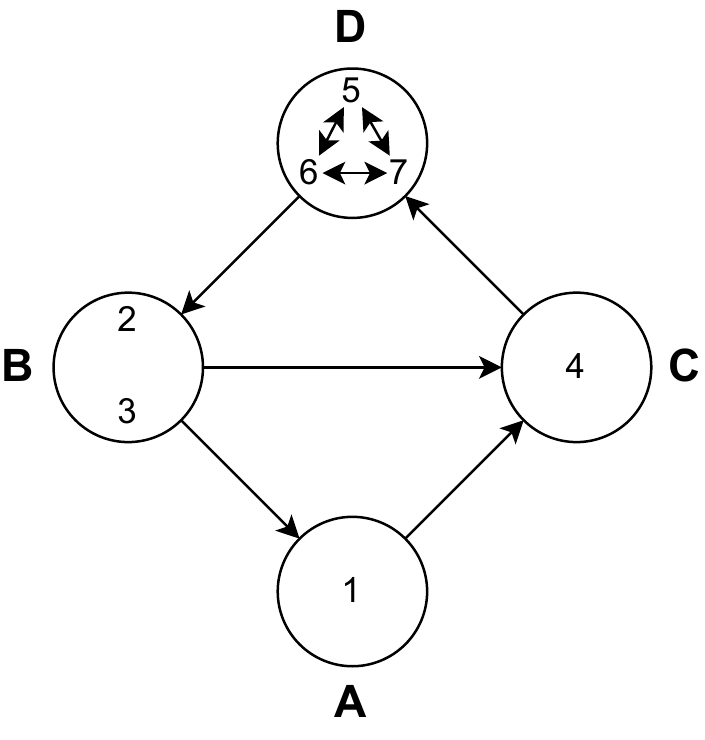}
        \caption{$\mathcal{P}_7^{(2)}$ as a directed graph.}
        \label{fig:7_unstable_2}
    \end{subfigure}
    \caption{The three families of cycle unstable binary preferences for $n \leq 7.$}
    \label{fig:5_7_unstable}
\end{figure}

\subsubsection{Two-Valued Preferences on Cycles.}\label{sec:2_valued_cycle}

Cycles being regular graphs, it is sufficient to study binary preferences. We exhausted all binary preferences for $n \leq 7$ using a Z3 Python solver (see Appendix \ref{app:Z3_solver}). Unstable preferences were found for $n=5$ and $n = 7,$ with one and, respectively, two non-isomorphic families of unstable preferences (see Table \ref{tab:summary_bin_results}). Examples of such preferences from each family $\mathcal{P}_5, \mathcal{P}_7^{(1)}$ and $\mathcal{P}_7^{(2)}$ are illustrated in Figure \ref{fig:5_7_unstable}.
Analyzing why $\mathcal{P}_5$ is unstable turns out to be quite complex (see Appendix \ref{sec:p5_stability_analysis}), and, as of our current understanding, its instability seems to be more of a ``small size artifact'' than anything else. In contrast, with their highly regular structure, the two instances with $n = 7$ seem more promising. In particular, profile $\mathcal{P}_7^{(2)}$ consists of only four classes, denoted by $A, B, C, D$ in Figure \ref{fig:7_unstable_2}. In the following, we show that $\mathcal{P}_7^{(2)}$ can be extended to unstable preferences of any $n \geq 7.$

\begin{reptheorem}{thm:bin_unst_cycle}
For $n\ge 7$, there exist four-class binary preferences such that all arrangements on a cycle are unstable. 
\end{reptheorem}
\begin{proof}
For $n\geq 7$, we consider four classes $A$, $B$, $C$ and $D$, as well as their respective members $a$, $b_1$, $b_2$, $c$ and $d_1,\ldots,d_{n-4}.$ Similarly to Figure \ref{fig:7_unstable_2}, we suppose that: (i) $a$ only likes $c$; (ii) $b_1$ and $b_2$ both only like $a$ and $c$; disliking each other; (iii) $c$ only likes members of $D$; (iv) members of $D$ all like each other, as well as $b_1$ and $b_2$, only disliking $a$ and $c.$ We show in the appendix why such preferences induce no stable arrangements on a cycle.
\end{proof}

This result shows that four classes of agents are sufficient to make all arrangements unstable on a cycle for two-valued preferences; we show in the following it is also necessary, as three classes of agents with two-valued preferences always induce a stable arrangement on a cycle.

\begin{reptheorem}{thm:3_class_2_valued_cycle_stable}
Three-class two-valued preferences always induce a stable arrangement on a cycle.
\end{reptheorem}
\begin{proof}
We consider three classes:~\textcolor{Red}{Reds}, \textcolor{LimeGreen}{Greens}, and \textcolor{blue}{Blues}, each containing $r$, $g$, and $b$ agents respectively. Without loss of generality, assume that $r, g, b \geq 1.$ Since the seating graph is regular, recall that we may assume the preferences to be binary. Note moreover that any blocking pair must have two agents of different colors. A principled case distinction now allows us to relatively quickly exhaust over all the possibilities for the preferences. We present the case distinction here and delegate the proofs themselves to specialized lemmas in the appendix.

First, whenever at least one class likes its own kind, Lemmas \ref{lem:3_class_sta:2_class_narciss}, \ref{lem:3_class_sta:1_class_likes_itself+other} and \ref{lem:3_class_sta:1_class_only_narciss} together show the existence of a stable arrangement. Note that the proofs are constructive and all the stable arrangements presented intuitively seat the self-liking class consecutively.

Then, if no class likes itself, Lemmas  \ref{lem:3_class_sta:no_narciss_with_paria} and \ref{lem:3_class_sta:no_narciss_1_like_all} show the existence of a stable arrangement whenever one class likes or is disliked by every other class. Note this time that intuitively all constructions present an alternation of two classes.

This only leaves us to handle the case where each class likes and is liked by exactly another class:~this case is treated separately in Lemma \ref{lem:3_class_sta:max_util_arg}, where we show that an arrangement maximizing utilitarian social welfare is always stable. The proof also shows that a modified variant of the swap dynamics converges.
\end{proof}

\subsubsection{Two-Valued Preferences on Paths.}\label{sec:2_valued_path} 
We now consider the more complex case of paths, where the endpoints and the associated loss of regularity artificially introduce an implicit comparison with zero, hence giving rise to a change of behavior between positive and negative preferences. In general, if we allow for negative values in the preferences, then there exist three-class two-valued preferences such that no arrangement on a path is stable:

\begin{reptheorem}{thm:3_class_2_val_unst_path}
For $n\ge 3$, there exist three-class two-valued preferences such that all arrangements on a path are unstable.
\end{reptheorem}
\begin{proof}
    Consider a variant of ``Alice, Bob and Friends'' where the preferences of Bob towards Alice, Alice towards Friends, and Friends towards Bob and themselves are all one; the preferences of Alice towards Bob, Bob towards Friends, and Friends towards Alice are all minus one. We show in the appendix why all arrangements on a path are unstable.
\end{proof}

For non-negative preferences, on the other hand, exhaustion for $n \leq 7$ using a similar solver\footnote{One can show that it suffices to try the cases $\Gamma \in \{\{0, 1\}, \{1, 2\}, \{1, 3\}, \{2, 3\}\}.$}  yields no unstable instances (see Table \ref{tab:summary_bin_results}). Surprised by the outcome, we also wrote \textsc{C++} code to test all $k$-class instances with at most $b$ agents per class for $(k, b) \in \{(4, 10), (5, 4)\},$ also leading to no unstable instances. \emph{We conjecture that two-valued instances with non-negative preferences always induce a stable arrangement on a path.} In the following, we show that this is true under the additional assumption that two agents are only willing to swap seats when they are at most two positions away on the path, no matter how much their utilities would increase otherwise.
This can be thought of as a practical constraint:~once the agents are seated, each agent knows which other agents they envy, but finding out whether envy is reciprocal would be too cumbersome if the other agent is seated too far away. 
For this setup, we prove that the swap dynamics always converge, so a stable arrangement can be found by starting with an arbitrary arrangement and swapping blocking pairs until the arrangement becomes stable. This can be seen as a generalization of a result from \cite{bilo_nash_2018}, where agents only have preference for others of their own kind and swaps are only with adjacent agents. This is stated below and proven next.

\begin{theorem}\label{thm:close-swaps-two-valued-is-stable} 
Two-valued non-negative preferences always induce a stable arrangement on a path assuming that agents are only willing to exchange seats with other agents sitting at distance at most two on the path. Moreover, the swap dynamics converge in this case.
\end{theorem}

To begin showing this, note that no local swap could occur with an agent seated at either endpoint, since the preferences are non-negative. 
Since the two endpoints are the only irregularities, removing them from consideration, we can restrict our analysis from non-negative two-valued to binary without loss of generality. 
For any arrangement $\pi,$ define the utilitarian social welfare $W(\pi) = \sum_{i \in \mathcal{A}}U_i(\pi).$ Moreover, to each arrangement $\pi$ we associate a sequence $S(\pi)$ of length $n - 1$ with elements in $\{0, 1, 2, 3\},$ constructed as follows. Let $\pi_i$ and $\pi_{i + 1}$ be the agents sitting at positions $i$ and $i + 1$ on the path:~if they do not like each other, then $S(\pi)_i = 0$, if they both like each other, then $S(\pi)_i = 3,$ if only $\pi_i$ likes $\pi_{i + 1},$ then $S(\pi)_i = 1,$ otherwise $S(\pi)_i = 2.$ To prove that the swap dynamics converge, we define the potential $\Phi(\pi) = (W(\pi), S(\pi)),$ where sequences are compared lexicographically, and prove that swapping blocking pairs always strictly increases the potential. The following two lemmas show this for swaps at distances one and two, respectively.

\begin{lemma} Let $\pi$ be an arrangement where $a$ and $b$ form a blocking pair and sit in adjacent seats. Let $\pi'$ be $\pi$ with $a$ and $b$'s seats swapped. Then, $\Phi(\pi') > \Phi(\pi).$
\end{lemma}

\begin{proof} When $n \leq 2,$ there are no blocking pairs, so assume $n \geq 3.$ First, notice that swapping the places of $a$ and $b$ keeps $a$ and $b$ adjacent, from which the swap changes the utility of any agent by at most one. Since $a$ and $b$'s utilities have to increase, they have to each change by exactly one. Moreover, note that neither $a$ nor $b$ can be seated at the ends of the table, as otherwise swapping would make one of them lose a neighbor while keeping the other one, hence not increasing their utility. Hence, assume that $x$ is the other neighbor of $a$ and $y$ is the other neighbor of $b$; i.e., $x, a, b, y$ are seated consecutively in this order on the path, at positions say $i, \ldots, i + 3.$ If either $U_x(\pi') \geq U_x(\pi)$ or $U_y(\pi') \geq U_y(\pi),$ it follows that $W(\pi') - W(\pi) 
= U_x(\pi') - U_x(\pi) + U_y(\pi') - U_y(\pi) + 2 \geq 1,$ so $\Phi(\pi') > \Phi(\pi).$ Otherwise, we know that $U_x(\pi') - U_x(\pi) = U_y(\pi') - U_y(\pi) = -1,$ from which $W(\pi') = W(\pi).$ Together with $U_a(\pi') - U_a(\pi) = U_b(\pi') - U_b(\pi) = 1,$ this means that preferences satisfy $a \rightarrow y \rightarrow b \rightarrow x \rightarrow a,$ where an arrow $u \rightarrow v$ indicates that agent $u$ likes $v$ but not the other way around. Therefore, $S(\pi)_i = 1$ and $S(\pi')_i = 2.$ Since $S(\pi)$ and $S(\pi')$ only differ at positions $i, \ldots, i + 2,$ this means that $S(\pi') > S(\pi),$ so $\Phi(\pi') > \Phi(\pi),$ as required.
\end{proof}

\begin{lemma} Let $\pi$ be an arrangement where $a$ and $b$ form a blocking pair and sit two seats away. Let $\pi'$ be $\pi$ with $a$ and $b$'s seats swapped. Then, $\Phi(\pi') > \Phi(\pi).$
\end{lemma}
\begin{proof} The same argument works, except that now we consider five agents $x, a, z, b, y$ seated at positions $i, \ldots, i + 4.$ This is because agent $z$ remains a common neighbor to $a$ and $b$ when swapping places, and can, essentially, be ignored.
\end{proof}

Therefore, since the potential is upper-bounded,
we get that the swap dynamics have to converge. One might now rightfully ask whether convergence is guaranteed to take polynomial time. While we could neither prove nor disprove this, in Appendix \ref{sec:exponential_convergence} we give evidence of why exponential time might be required. Moreover, note that for cycles convergence is not guaranteed even for swaps at distance at most two; e.g., $\mathcal{P}_5$ in Figure \ref{fig:5_unst}, where any two agents are seated at most two seats away anyway. For paths, on the other hand, one could still hope that the result generalizes beyond distance at most two when the preferences are non-negative. This is however not the case, even when stable arrangements exist, as we show next (details in the appendix).

\begin{figure}[t]
    \centering
    \begin{subfigure}{0.45\textwidth}
        \centering
        \includegraphics[width=3.5cm]{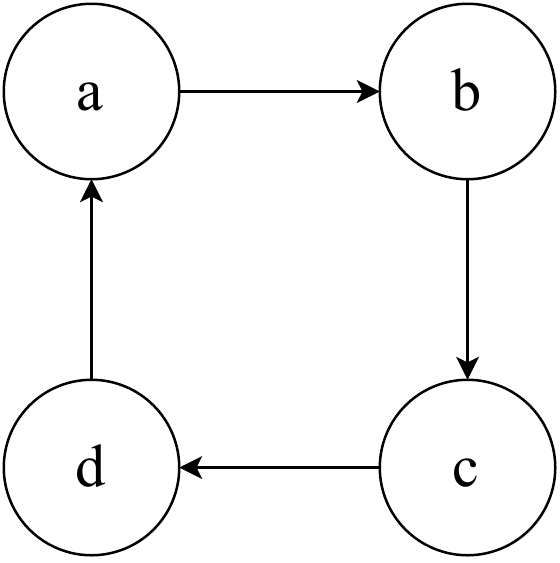}
        \caption{$\mathcal{P}_4$ as a directed graph.}
        \label{fig:4_non_cv_pref}
    \end{subfigure}
    \hfill
    \begin{subfigure}{0.45\textwidth}
        \centering
        \includegraphics[width=5.5cm]{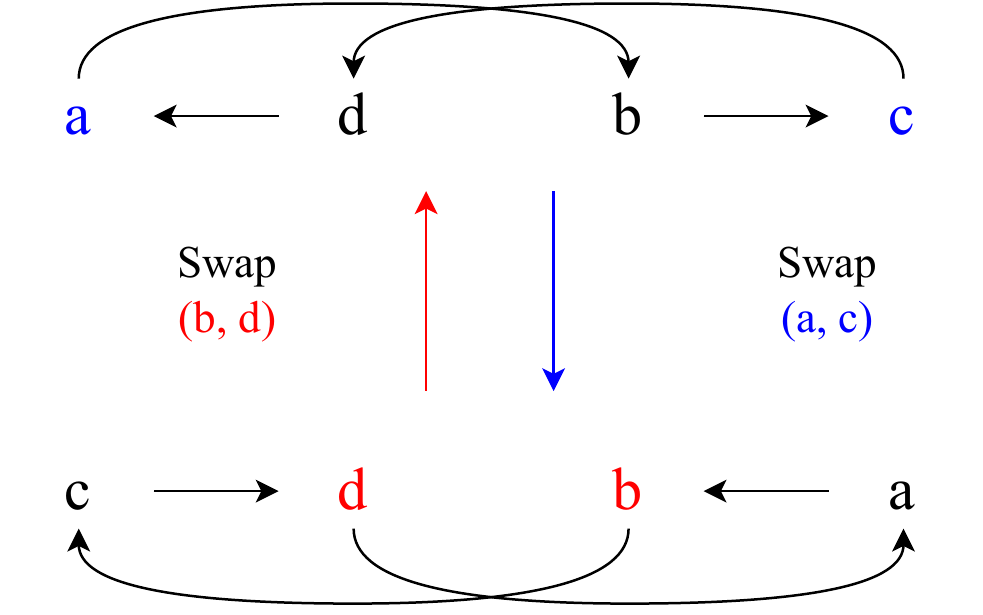}
        \caption{Looping swap dynamics.}
        \label{fig:4_non_cv_exple}
    \end{subfigure}
    \caption{Four-agent binary preferences $\mathcal{P}_4$ with path stable arrangement $\pi^* = (a, b, c, d)$ but where the swap dynamics
    necessarily alternate between $\pi_1 = (a, d, b, c)$ and $\pi_2 = (c, d, b, a),$ up to reversal, when started in either of them.%
    }
    \label{fig:4_non_cv}
\end{figure}

\begin{replemma}{lem:4_non_cv}
Consider profile $\mathcal{P}_4$ in Figure \ref{fig:4_non_cv}. A path stable arrangement exists, yet the swap dynamics started from certain arrangements cannot converge.
\end{replemma}

This generalizes to  $n \geq 4$ agents by adding $n - 4$ dummy agents to $\mathcal{P}_4$ liking nobody and being liked by nobody and seating them at positions $5, \ldots, n$ on the path. Hence, non-convergence for distance $\leq 
3$ is not a small-$n$ artifact.

\section{Conclusions and Future Work}
We studied envy-freeness and exchange-stability on paths and cycles. For both topologies, we showed that finding envy-free/stable arrangements can be achieved in polynomial time when the number of agent classes is bounded, while for envy-freeness the problem becomes NP-hard without this restriction, even for binary preferences. For stability, it is known that for sufficiently many values the problem is also NP-hard \cite{hua}. However, it would be interesting to see, for instance, if the same can be said about binary preferences. For cycles at least, we believe this to be the case, but were unable to prove it. In part, this is because of the difficulty of constructing unstable instances in the first place. Moreover, for both topologies, we gave a full characterization of the pairs $(k, v)$ such that $k$-class $v$-valued unstable preferences exist. 
For paths, the characterization requires negative values in the two-valued case, and we are still unsure whether two-valued non-negative preferences that are unstable on a path exist. We, however, partially answer this in the negative by showing that the swap dynamics are guaranteed to converge if agents can only swap places with other agents seated at most two positions away from them. Without this assumption, convergence might not be guaranteed even when stable arrangements exist, so a different approach would be required to prove existence. It would also be interesting to know if unstable preferences are exceptions or the norm. We give a probabilistic treatment of this question for random preference digraphs of average degree $O(\sqrt{n})$ in the appendix. As an avenue for future research, it would be attractive to consider other kinds of tables commonly used in practice, the most relevant being the one shaped as a $2 \times n$ grid, with guests on either side facing each other.
It would additionally be interesting to consider non-additive utilities or to enforce additional constraints on the arrangement, such as certain people coming ``in groups'' and hence having to sit consecutively at the table. It would also be worth investigating replacing addition by taking minimum in the agents' utilities (and also the more general lexicographical variant).

\subsubsection{Acknowledgements} We thank Edith Elkind for the useful discussions regarding this work. We thank Giovanna Varricchio and Martin Bullinger for pointing out relevant literature. A preliminary version of this paper was presented at COMSOC'23; we thank the workshop attendees for their interesting questions and remarks. We thank the anonymous reviewers for their constructive feedback and useful suggestions contributing to improving the paper.


\bibliographystyle{splncs04}
{\footnotesize 
\bibliography{wine}

\begin{thebibliography}{10}
\providecommand{\url}[1]{\texttt{#1}}
\providecommand{\urlprefix}{URL }
\providecommand{\doi}[1]{https://doi.org/#1}

\bibitem{elkind_schelling_2021}
Agarwal, A., Elkind, E., Gan, J., Igarashi, A., Suksompong, W., Voudouris,
  A.A.: Schelling games on graphs. Artificial Intelligence  \textbf{301},
  103576 (2021)

\bibitem{handbook_hedonic}
Aziz, H., Savani, R.: Hedonic games. In: Brandt, F., Conitzer, V., Endriss, U.,
  Lang, J., Procaccia, A.D. (eds.) Handbook of Computational Social Choice,
  chap.~15. Cambridge University Press, USA, 1st edn. (2016)

\bibitem{bilo_topological_2020}
Bilò, D., Bilò, V., Lenzner, P., Molitor, L.: Topological influence and
  locality in swap schelling games. Autonomous Agents and Multi-Agent Systems
  \textbf{36} (08 2022)

\bibitem{bilo_nash_2018}
Bilò, V., Fanelli, A., Flammini, M., Monaco, G., Moscardelli, L.: Nash
  {Stable} {Outcomes} in {Fractional} {Hedonic} {Games}: {Existence},
  {Efficiency} and {Computation}. Journal of Artificial Intelligence Research
  \textbf{62},  315--371 (Jun 2018)

\bibitem{bilo_hedonic_2022}
Bilò, V., Monaco, G., Moscardelli, L.: Hedonic {Games} with {Fixed}-{Size}
  {Coalitions}. AAAI'22  \textbf{36}(9),  9287--9295 (Jun 2022), number: 9

\bibitem{bodlaender_hedonic_2020}
Bodlaender, H.L., Hanaka, T., Jaffke, L., Ono, H., Otachi, Y., van~der Zanden,
  T.C.: Hedonic {Seat} {Arrangement} {Problems}. arXiv  (Feb 2020)

\bibitem{run_and_chase}
Boehmer, N., Bullinger, M., Kerkmann, A.M.: Causes of stability in dynamic
  coalition formation. AAAI'23  \textbf{37}(5),  5499--5506 (Jun 2023)

\bibitem{bullinger_topological_2022}
Bullinger, M., Suksompong, W.: Topological distance games. AAAI'23
  \textbf{37}(5),  5549--5556 (Jun 2023)

\bibitem{cechlarova_complexity_2002}
Cechlárová, K.: On the complexity of exchange-stable roommates. Discrete
  Applied Mathematics  \textbf{116}(3),  279--287 (Feb 2002)

\bibitem{hua}
Ceylan, E., Chen, J., Roy, S.: Optimal seat arrangement: What are the hard and
  easy cases? In: Elkind, E. (ed.) IJCAI'23. pp. 2563--2571 (8 2023)

\bibitem{chauhan_schelling_2018}
Chauhan, A., Lenzner, P., Molitor, L.: Schelling segregation with strategic
  agents. In: Deng, X. (ed.) Algorithmic Game Theory. pp. 137--149. Springer
  International Publishing, Cham (2018)

\bibitem{gale_college_1962}
Gale, D., Shapley, L.S.: College {Admissions} and the {Stability} of
  {Marriage}. The American Mathematical Monthly  \textbf{69}(1),  9--15 (Jan
  1962)

\bibitem{irving_efficient_1985}
Irving, R.W.: An efficient algorithm for the “stable roommates” problem.
  Journal of Algorithms  \textbf{6}(4),  577--595 (Dec 1985)

\bibitem{irving_stable_marriage_indifference}
Irving, R.W.: Stable marriage and indifference. Discrete Applied Mathematics
  \textbf{48}(3),  261--272 (1994)

\bibitem{kreisel_equilibria_2022}
Kreisel, L., Boehmer, N., Froese, V., Niedermeier, R.: Equilibria in schelling
  games: Computational hardness and robustness. In: AAMAS'22. pp. 761--769
  (2022)

\bibitem{ronn_np-complete_1990}
Ronn, E.: {NP}-complete stable matching problems. Journal of Algorithms
  \textbf{11}(2),  285--304 (Jun 1990)

\end{thebibliography}
}

\newpage
\appendix

\section{Omitted Proofs}

In this appendix, we provide the proofs omitted from the main text of the paper.

\subsection{Proofs Omitted From Section \ref{sec:envy-free}}

In this section, we prove Theorem \ref{th_np_hard_path}, restated below for convenience.

\repeattheorem{th_np_hard_path}
\begin{proof}
We proceed similarly as for Theorem \ref{th_np_hard_cycle}, this time reducing from Hamiltonian Path on directed graphs. To make the reduction work, we require the additional stipulation that the input graph $G$ has a vertex with no outgoing edges, which we assume without loss of generality to be vertex $n.$ Note that this preserves NP-hardness. Moreover, in the preprocessing step, will now only check that vertices $v \in V \setminus \{n\}$ have outgoing edges. Otherwise, the construction of preference profile $\calP$ stays the same. To show that if $G$ has a Hamiltonian path, then there exists an envy-free arrangement on a path, the same argument as before can be used, except for the treatment of $z_n,$ who gets a utility of zero, but they still can not envy another agent because they approve of nobody (all other agents retain utility 1, as before). To show that an envy-free arrangement on a path implies the existence of a Hamiltonian path, the argument stays similar but again requires minor tweaks. In particular, for agent $z_n$, and only for them, it holds that they do not necessarily need a second neighbor other than $y_n$, because they like no other agents, so they will be happy with a utility of 0, obtained by sitting at one of the ends of the path. The analysis for the other agents stays the same, as they do require a second neighbor to get a utility of 1, and hence cannot sit at the ends.
\end{proof}

\subsection{Proofs Omitted From Section \ref{sec:k_classes_algo}}

In this section, we prove Theorem \ref{thm:algo_k_class_cycle}, restated below for convenience.

\repeattheorem{thm:algo_k_class_cycle} 
\begin{proof} The proof idea stays similar to that for paths, first designing a nondeterministic logarithmic space algorithm deciding whether a compatible arrangement exists and then lifting this to one running in deterministic polynomial time. Since cycles no longer have endpoints, the definition of compatible arrangements needs adjusting. First, instead of introducing a dummy agent class 0 and placing it at positions $0$ and $n + 1$ in s, we now make $s_0$ stand for $s_n$, and $s_{n + 1}$ stand for $s_1.$ Similarly, $s_{n + 2} = s_{2},$ etc. Moreover, in order for arrangement $s$ to be compatible, we now require that for all $1 \leq i < j \leq n$ such that $1 < j - i < n - 1,$ the triplets $(s_{i - 1}, s_i, s_{i + 1})$ and $(s_{j - 1}, s_j, s_{j + 1})$ are long-range compatible and, additionally, for all $1 \leq i \leq n$ the triplets $(s_{i - 1}, s_i, s_{i + 1})$ and $(s_{i}, s_{i + 1}, s_{i + 2})$ are short-range compatible. 
Note, therefore, that the pair $(i, j) = (1, n)$ is the only pair for which the required check differs from the path case:~previously, the check was for long-range compatibility, while now it is for short-range compatibility, and technically with the roles of $i$ and $j$ reversed. This time, instead of beginning the algorithm by guessing the values of $s_1$ and $s_2,$ we instead begin it by guessing the values of $s_0, s_1$ and $s_2.$ The algorithm then proceeds as before. However, when the value $i = n$ is reached, naturally, no guessing takes place, as $s_n$ has already been guessed, but all other computations execute as before. However, when $i = n + 1$ is reached, it is trickier; not only do we not have to guess $s_{n + 1} = s_1,$ but also the check performed has to be altered. In particular, instead of checking whether $(s_{n - 1}, s_n, s_1)$     is long-range compatible with $(s_n, s_1, s_2),$ the check now has to be for short-range compatibility. Checking for short-range compatibility can easily be incorporated, so it remains to show how to ensure that the two triplets are not also tested for long-range compatibility like in the previous implementation. This is done as follows:~instead of temporarily decreasing the values $m_{s_{i - 3}, s_{i - 2}, s_{i - 1}}$ and $m_{s_{i - 2}, s_{i - 1}, s_i}$ by one and then checking $(s_{i - 2}, s_{i - 1}, s_i)$ against the counts in $m$, we now do the same but also decrease $m_{s_{n}, s_{1}, s_2}$ by one. We stress that these rather tedious modifications are only applied for $i = n + 1.$ The modified algorithm needs to store $s_0, s_1$ and $s_2$ throughout its execution in addition to the state it already stored, but this does not impact the logarithmic space-bound, completing the proof.
\end{proof}

\subsection{Proofs Omitted From Section \ref{sec:2_classes}}

In this section, we prove Theorem \ref{thm:2_class_path_stable}, restated below for convenience.

\repeattheorem{thm:2_class_path_stable} 
\begin{proof}
    The introduction of agents sitting at the two endpoints calls for a more careful analysis:
    \smallbreak
    \begin{itemize}
        \item If one class, say \textcolor{blue}{Blue}, likes everyone equally.
        
        Then \textcolor{blue}{Blues} either prefer to be on the endpoint or in the middle of the path (depending on the sign of their constant preference). If they prefer endpoints, seating two \textcolor{blue}{Blues} on the endpoints makes the arrangement stable; otherwise, sitting all \textcolor{blue}{Blues} in the middle ensures stability.
        
        Those two solutions are not available if and only if there is a single agent in one of the two classes:~in that case, ensuring that this agent gets maximal utility stabilizes the whole arrangement.
        
    \item If \textcolor{blue}{Blues} and \textcolor{Red}{Reds} prefer the same class, say \textcolor{blue}{Blues}.
    
        The reasoning from the previous proof still holds, as long as no \textcolor{blue}{Blue} sits at an endpoint. This would happen if and only if there is a unique \textcolor{Red}{Red}: ensuring maximum utility for that one agent would then give stability.

    \item If \textcolor{blue}{Blues} and \textcolor{Red}{Reds} both strictly prefer the opposite kind, alternate \textcolor{blue}{Blues} and \textcolor{Red}{Reds} starting from the most numerous class, say \textcolor{blue}{Blue}.
    Suppose:
    \begin{itemize}
        \item Preferences for the opposite kind are positive for both. Whenever there is strictly more of one class than the other, the reasoning from the previous proof still holds. In case of equality, only the \textcolor{blue}{Blue} and \textcolor{Red}{Red} at the endpoints would want to swap:~it is not the case, as they would exchange a ``different colour'' neighbor for a ``same colour'' one.
        
        \item Preferences for the opposite kind are negative for both. Now at most two agents have maximum utility:~the \textcolor{blue}{Blue} extremal one, and the \textcolor{Red}{Red} extremal one if there are as many \textcolor{Red}{Reds} as \textcolor{blue}{Blues}. If both have maximum utility, the others cannot improve, and the arrangement is stable. The arrangement is still stable if there are strictly more \textcolor{blue}{Blues} than \textcolor{Red}{Reds}, as the extremal \textcolor{blue}{Blues} could only agree to swap between two \textcolor{Red}{Reds}. As only \textcolor{blue}{Blues} are sitting between two \textcolor{Red}{Reds}, this cannot happen.
        
        \item Preference of \textcolor{blue}{Blues} for \textcolor{Red}{Reds} is negative, but \textcolor{Red}{Reds} for \textcolor{blue}{Blues} is positive. Then non-extremal \textcolor{blue}{Blues} would envy the extremal \textcolor{Red}{Red}, which would envy them back if and only if the preference of \textcolor{Red}{Reds} towards \textcolor{Red}{Reds} is strictly positive. In this latter case, proceed to the exchange:~\textcolor{blue}{Blues} in the middle can only improve by switching with an extremal \textcolor{blue}{Blue}, which would never be accepted. This new arrangement is therefore stable.
    \end{itemize}
        
     \item If \textcolor{blue}{Blues} and \textcolor{Red}{Reds} both strictly prefer their own kind, sit all \textcolor{blue}{Blues} on one side and all \textcolor{Red}{Reds} on the other. 
     Let $\textcolor{blue}{B_{\mathit{out}}}$ be the extremal \textcolor{blue}{Blue}, $\textcolor{blue}{B_{\mathit{in}}}$ be the only \textcolor{blue}{Blue} with both a \textcolor{blue}{Blue} and a \textcolor{Red}{Red} neighbor $\textcolor{Red}{R_{\mathit{in}}}$, and let $\textcolor{Red}{R_{\mathit{out}}}$ be the extremal \textcolor{Red}{Red}. 
     As previously, $(\textcolor{blue}{B_{\mathit{in}}}, \textcolor{Red}{R_{\mathit{in}}})$ is never a blocking pair. Suppose:
     
     \begin{itemize}
         \item Preferences for their own kind are positive for both. We verify that $(\textcolor{blue}{B_{\mathit{in}}}, \textcolor{Red}{R_{\mathit{out}}})$ is not a blocking pair, as $\textcolor{blue}{B_{\mathit{in}}}$ would lose his only \textcolor{blue}{Blue} neighbor. By symmetry, the only remaining possible blocking pair is $(\textcolor{blue}{B_{\mathit{out}}}, \textcolor{Red}{R_{\mathit{out}}})$:~it is also not a blocking pair since both would only gain an ``opposite colour'' neighbor.
         \item Preferences for their own kind are negative for both. Both  $\textcolor{blue}{\textcolor{blue}{B_{\mathit{out}}}} $ and $\textcolor{Red}{R_{\mathit{out}}}$ have maximum utility, hence $\textcolor{blue}{B_{\mathit{out}}}$ an $\textcolor{Red}{R_{\mathit{out}}}$ are not part of a blocking pair. Other \textcolor{blue}{Blues} except $\textcolor{blue}{B_{\mathit{in}}}$ can only improve their utility by moving to an endpoint, i.e., by switching with $\textcolor{Red}{R_{\mathit{out}}}$ or $\textcolor{blue}{B_{\mathit{out}}}$, both being impossible. The only remaining possibility $(\textcolor{blue}{B_{\mathit{in}}}, \textcolor{Red}{R_{\mathit{in}}})$ is also not a blocking pair. Therefore, the arrangement is stable.
         \item Preference of \textcolor{blue}{Blues} for \textcolor{blue}{Blues} is negative, but that of \textcolor{Red}{Reds} for \textcolor{Red}{Reds} is positive. Note that all \textcolor{Red}{Reds} but $\textcolor{Red}{R_{\mathit{in}}}$ and $\textcolor{Red}{R_{\mathit{out}}}$ have maximum utility, hence cannot be part of a blocking pair. Moreover, since $\textcolor{blue}{B_{\mathit{out}}}$ has maximum utility, $\textcolor{Red}{R_{\mathit{in}}}$ could only switch with \textcolor{blue}{Blues} having at least one blue neighbor, so it cannot improve. 
         If \textcolor{Red}{Reds} dislike \textcolor{blue}{Blues}, $\textcolor{Red}{R_{\mathit{out}}}$ is also unable to improve, and the arrangement is stable. 
         On the contrary, suppose \textcolor{Red}{Reds} strictly like \textcolor{blue}{Blues}, and consider the arrangement obtained after exchanging $\textcolor{Red}{R_{\mathit{out}}}$ and $\textcolor{blue}{B_{\mathit{in}}}$:~now both extremal \textcolor{Red}{Reds} have the second highest utility, and cannot improve since no \textcolor{blue}{Blue} is sitting between two \textcolor{Red}{Reds}, so the arrangement is stable.
     \end{itemize}
    \end{itemize}
    Hence, it is always possible to sit two classes in a stable manner on a path.
\end{proof}

\subsection{Proofs Omitted From Section \ref{sec:3-val_3-class}}

In this section, we prove Theorem \ref{thm:unst_tern_path}, restated below for convenience.

\repeattheorem{thm:unst_tern_path}

\begin{figure}[t]
$$ \mathcal{P} =\left[\begin{array}{ c  c  c  c| c  c  c}
    0 & 0 & 0 & 0 & 1 & \ldots & 1\\
    1 & 0 & 0 & 0 & 0 & \ldots & 0\\
    1 & 0 & 0 & 0 & 0 & \ldots & 0\\
    1 & 0 & 0 & 0 & 0 & \ldots & 0\\
    \hline
    0 & 3 & 3 & 3 & 1 & \ldots & 1\\
    \ldots & \ldots & \ldots & \ldots & \ldots & \ldots & \ldots\\
    0 & 3 & 3 & 3 & 1 & \ldots & 1\\
  \end{array}
  \right]$$
\caption{Possible instance of preferences for the proof of Theorem \ref{thm:unst_tern_path}, with the ordering (Alice, $1\textsuperscript{st}$  Bob, $2\textsuperscript{nd}$ Bob, $3\textsuperscript{rd}$ Bob, Friends).}
\label{fig:exple_pref_unst_3values_path}
\end{figure}

\begin{proof}
We consider the same instance of preferences as in the proof of Theorem \ref{thm:unst_ternary_cycle}, with the following modifications:~we add two copies of Bob, and we suppose there are at least eight friends. Alice likes everyone but the Bobs, the Bobs only like Alice, the friends like the Bobs the most and Alice the least. We furthermore suppose that friends would rather be next to one Bob than between two other friends. Figure \ref{fig:exple_pref_unst_3values_path} displays a possible instance of such preferences.

Suppose one of the Bobs is sitting next to Alice. Since there are at least eight different friends, at least one of them is neither sitting beside a Bob nor at an endpoint of the table. Indeed, at most five friends are sitting next to Bobs (since Alice sits next to one of them), and two more friends can be sitting at an endpoint. Hence, Alice and this friend would both agree to switch places since it is always worth it for a friend to move beside a Bob, even if this means sitting at an endpoint.

Now, suppose no Bob is sitting next to Alice. Since there are three Bobs, at least one of them is not sitting at an endpoint, say $B_1$. The only case where one of Alice's neighbors would not agree to switch with $B_1$ is if it was already sitting next to another Bob, say $B_2$; moreover, the only reason for him not to switch with $B_2$ is if $B_2$ is sitting at the endpoint of the table. Hence, the seating arrangement is of the form $(B_2,~F,~A,~F,~\ldots,~B_1,~\ldots).$ Alice is not sitting at an endpoint, therefore possesses a second friend as a neighbor, and that second friend would agree to swap with at least one of the two remaining Bobs.
\end{proof}

\subsection{Proofs Omitted From Section \ref{sec:2_valued_cycle}}

In this section, we complete the proof of Theorem \ref{thm:bin_unst_cycle}. Then, we prove the lemmas used in the proof of Theorem \ref{thm:3_class_2_valued_cycle_stable}. Subsequently, we complete the proof of Theorem \ref{thm:3_class_2_val_unst_path}. Finally, we prove Lemma \ref{lem:4_non_cv}.

\subsubsection{Proving Theorem \ref{thm:bin_unst_cycle}.} We now complete the proof of Theorem \ref{thm:bin_unst_cycle}, restated below for convenience. 

\repeattheorem{thm:bin_unst_cycle}
\begin{proof}[continued]
We now show that all arrangements on a cycle are unstable for our preference profile, by considering every possible local arrangement around $c$ and showing that they all induce a blocking pair.

\begin{itemize}
    \item If the local arrangement around $c$ consists of $(b_1,~c, ~b_2)$, then $c$ has utility 0 and would agree to swap with anyone having a neighbor in $D.$ Since $D$ contains strictly more than two agents, one of the $d_i$ neighbors of $a$ has another member of $D$ as a neighbor. Performing a swap with $c$ would increase its utility from 1 to 2, while $c$ would improve to utility 1, so it is a blocking pair.
    
    \item If it consists of $(b_1, ~c, ~a),$ $c$ has again utility 0 and would switch with whoever has a neighbor in $D.$ In particular, if the local arrangement is $(b_1, ~c, ~a, ~b_2)$, then $(b_1,c)$ is a blocking pair, as $b_1$ second neighbor is in $D.$ If it is $(b_2, ~b_1, ~c, ~a)$, then $(b_2,c)$ is a blocking pair, as $b_2$'s second neighbor is in $D.$ Otherwise, the local arrangement must be $(d_i, ~b_1, ~c, ~a, ~d_j)$, and $(b_2,c)$ forms once again a blocking pair. The same reasoning naturally holds for local arrangements of the form $(b_2, ~c, ~a).$
    
    \item If the local arrangement around $c$ consists of $(d_i, ~c, ~a),$ then at least one member of $B$, say $b_1$, is no neighbor of $a$ and has utility 0. In that case, both $d_i$ and $b_1$ can increase their utility by exchanging seats.
    
    \item At last, if it consists of $(d_i, ~c, ~b_1)$ or $(d_i, ~c, ~d_j),$ then agent $a$ has utility 0, and $d_i$ can always increase its utility by switching with $a.$ Indeed, if $a$ was his neighbor, he would exchange $c$ for a member of $D \cup B$ while retaining $a$ as a neighbor; otherwise he would exchange $c$ for a second neighbor in $D \cup B.$ Since moving close to $c$ would always increase $a$'s utility, $(d_i, a)$ is a blocking pair. The same reasoning of course holds for local arrangements of the form $(d_i, ~c, ~b_2).$
\end{itemize}
\end{proof}

\subsubsection{Proving Theorem \ref{thm:3_class_2_valued_cycle_stable}.} We now prove the Lemmas used to prove Theorem \ref{thm:3_class_2_valued_cycle_stable}, restated below for convenience.

\repeattheorem{thm:3_class_2_valued_cycle_stable}

\begin{lemma}\label{lem:3_class_sta:2_class_narciss}
    Suppose at least two classes like their own kind. Then, there exists a stable arrangement.
\end{lemma}
\begin{proof}
Suppose \textcolor{Red}{Reds} like \textcolor{Red}{Reds} and \textcolor{blue}{Blues} like \textcolor{blue}{Blues}. Then, consider the arrangement where each class sits among themselves. Let $\textcolor{Red}{R_1}$ and $\textcolor{blue}{B_1}$ be the neighboring \textcolor{Red}{Red} and \textcolor{blue}{Blue} respectively, $\textcolor{Red}{R_2}$ (resp. $\textcolor{blue}{B_2}$) the \textcolor{Red}{Red} (resp. Blue) with a \textcolor{LimeGreen}{Green} neighbor (we can potentially have $\textcolor{Red}{R_1} =\textcolor{Red}{R_2}$). Note that all \textcolor{Red}{Reds} but $\textcolor{Red}{R_1}$ and $\textcolor{Red}{R_2}$ have maximum utility, $\textcolor{Red}{R_1}$ can only improve its utility by switching with a \textcolor{LimeGreen}{Green}, whereas $\textcolor{Red}{R_2}$ cannot improve by switching with a \textcolor{LimeGreen}{Green}. 
By symmetry, this also applies to $\textcolor{blue}{B_1}$ and $\textcolor{blue}{B_2}$, hence a blocking pair must include $\textcolor{Red}{R_1}$ (or $\textcolor{blue}{B_1}$) and a \textcolor{LimeGreen}{Green}. This constitutes a blocking pair only if $g > 1$, \textcolor{Red}{Reds} like \textcolor{LimeGreen}{Greens} but \textcolor{LimeGreen}{Greens} dislike themselves. If this is the case, separate the \textcolor{Red}{Reds} and \textcolor{blue}{Blues} by seating one \textcolor{LimeGreen}{Green} between them: all \textcolor{Red}{Reds} and inner \textcolor{blue}{Blues} now have maximum utility, and both extremal \textcolor{blue}{Blues} are now in $\textcolor{blue}{B_2}$'s previous case. Hence the arrangement is stable.
\end{proof}

\begin{lemma}\label{lem:3_class_sta:1_class_likes_itself+other}
    Suppose precisely one class likes its own kind, and that they like at least another class. Then, there exists a stable arrangement. 
\end{lemma}
\begin{proof}
    Suppose that \textcolor{Red}{Reds} like \textcolor{Red}{Reds} and \textcolor{blue}{Blues}, \textcolor{blue}{Blues} dislike \textcolor{blue}{Blues} and \textcolor{LimeGreen}{Greens} dislike \textcolor{LimeGreen}{Greens}. We start building a stable arrangement by seating all \textcolor{Red}{Reds} consecutively (on a path). Call the extremal \textcolor{Red}{Reds} the two (or unique if $r = 1$) \textcolor{Red}{Reds} at the endpoints of the path, and note that all \textcolor{Red}{Reds} but the extremal ones already have maximum utility. We then complete the arrangement starting from the neighbor of one of the extremal \textcolor{Red}{Reds}.
    \begin{itemize}
        \item If $g < b$, start with a \textcolor{blue}{Blue} and alternate a \textcolor{LimeGreen}{Green} and a \textcolor{blue}{Blue} until all \textcolor{LimeGreen}{Greens} are seated, giving the potential remaining seats to \textcolor{blue}{Blues}: \textcolor{Red}{Reds} have maximum utility and no \textcolor{blue}{Blue} would accept a swap with a \textcolor{LimeGreen}{Green} (not even a neighbor), hence the arrangement is stable.

        \item If $g \geq b$ and \textcolor{LimeGreen}{Greens} like \textcolor{blue}{Blues}, start with a \textcolor{LimeGreen}{Green} and alternate a \textcolor{blue}{Blue} and a \textcolor{LimeGreen}{Green} until all \textcolor{blue}{Blues} are depleted, giving the potential remaining seats to \textcolor{LimeGreen}{Greens}. Note that an extremal \textcolor{Red}{Red} would never accept a swap with a \textcolor{LimeGreen}{Green} unless it sits between two \textcolor{blue}{Blues}, in which case the \textcolor{LimeGreen}{Green} already has maximum utility. Moreover, a \textcolor{LimeGreen}{Green} would never accept a swap with a \textcolor{blue}{Blue}, even if it is a neighbor. Finally, since an exchange with a \textcolor{blue}{Blue} also cannot improve an extremal \textcolor{Red}{Red} utility, the arrangement is stable.

         \item If $g \geq b$ and \textcolor{LimeGreen}{Greens} dislike \textcolor{blue}{Blues}, start with a \textcolor{LimeGreen}{Green} and then alternate two \textcolor{blue}{Blues} and two \textcolor{LimeGreen}{Greens} until all \textcolor{blue}{Blues} are depleted, give the remaining seats to \textcolor{LimeGreen}{Greens}. In case there is a unique \textcolor{blue}{Blue}, seat him between two \textcolor{LimeGreen}{Greens}. Note that no swap with a \textcolor{LimeGreen}{Green} or a \textcolor{blue}{Blue} would give an extremal \textcolor{Red}{Red} two neighbors \textcolor{Red}{Red} or \textcolor{blue}{Blue}, hence would never be profitable. Furthermore, a swap with a \textcolor{blue}{Blue} would always give zero utility to a \textcolor{LimeGreen}{Green}, and would never be agreed upon. Hence the arrangement is stable.

    \end{itemize}
\end{proof}

\begin{lemma}\label{lem:3_class_sta:1_class_only_narciss}
Suppose precisely one class likes its own, and that they dislike both other classes. Then, there exists a stable arrangement.
\end{lemma}
\begin{proof}
    Suppose only \textcolor{Red}{Reds} like their own kind, but they also dislike both \textcolor{blue}{Blues} and \textcolor{LimeGreen}{Greens}. Seat all \textcolor{Red}{Reds} together and alternate \textcolor{blue}{Blues} and \textcolor{LimeGreen}{Greens} starting from the most numerous class, say \textcolor{blue}{Blue}. Note that no \textcolor{blue}{Blue} would agree to exchange with a \textcolor{LimeGreen}{Green} seating between two \textcolor{blue}{Blues}, even if they are neighbors and the \textcolor{blue}{Blue} has a \textcolor{Red}{Red} neighbor. Furthermore, in the case where $g =  b$, the same reasoning can be applied to all \textcolor{blue}{Blues} sitting between two \textcolor{LimeGreen}{Greens}, and the arrangement is stable because the \textcolor{blue}{Blue} and \textcolor{LimeGreen}{Green} with a \textcolor{Red}{Red} neighbor would not exchange seats. Hence the arrangement is stable.
\end{proof}

\begin{lemma}\label{lem:3_class_sta:no_narciss_with_paria}
    Suppose no class likes itself, and one class is disliked by the two others. Then, there exists a stable arrangement.
\end{lemma}
\begin{proof}
Suppose no class likes \textcolor{Red}{Reds}, even themselves.
Suppose without loss of generality that $ b \geq g$, start by seating all \textcolor{Red}{Reds} together, then alternate between a \textcolor{blue}{Blue} and a \textcolor{LimeGreen}{Green} starting from \textcolor{blue}{Blue} until everyone is seated. No \textcolor{LimeGreen}{Green} would accept to exchange seats with a \textcolor{blue}{Blue} sitting between two \textcolor{LimeGreen}{Greens}, even if they are neighbors. Furthermore, no such \textcolor{blue}{Blue} would benefit from switching with a \textcolor{Red}{Red}. By symmetry, it also holds for \textcolor{LimeGreen}{Greens}, and a blocking pair would have to contain one of the \textcolor{blue}{Blue}/\textcolor{LimeGreen}{Green} agents with a \textcolor{Red}{Red} neighbor: those would neither exchange seats between them nor with any \textcolor{Red}{Red}, hence the arrangement is stable.
\end{proof}

\begin{lemma}\label{lem:3_class_sta:no_narciss_1_like_all}
    Suppose no class likes itself, but one class likes the two others. Then, there exists a stable arrangement.
\end{lemma}
\begin{proof}
Suppose \textcolor{Red}{Reds} dislike their own kind but like the two other classes.
Using Lemma \ref{lem:3_class_sta:no_narciss_with_paria}, we may assume that at least one class, say \textcolor{blue}{Blues}, like \textcolor{Red}{Reds}.

\begin{itemize}
    \item If $r > b$, start by alternating \textcolor{Red}{Reds} and \textcolor{blue}{Blues} starting with \textcolor{Red}{Reds} until all \textcolor{blue}{Blues} are depleted, then alternate \textcolor{Red}{Reds} and \textcolor{LimeGreen}{Greens} starting with \textcolor{Red}{Reds} until one of them is depleted; seat the remaining agents altogether. Note that all \textcolor{blue}{Blues} have maximum utility, and if two \textcolor{LimeGreen}{Greens} are neighbors, then all \textcolor{Red}{Reds} also have maximum utility. Otherwise, all \textcolor{LimeGreen}{Greens} are seated between two \textcolor{Red}{Reds}, and no \textcolor{Red}{Red} can improve by swapping with a \textcolor{LimeGreen}{Green}. Hence the arrangement is stable.
    
    \item If $r < b$, start by alternating \textcolor{Red}{Reds} and \textcolor{blue}{Blues} starting with \textcolor{blue}{Blues} until all \textcolor{Red}{Reds} are depleted, then alternate \textcolor{blue}{Blues} and \textcolor{LimeGreen}{Greens} starting with \textcolor{blue}{Blues} until one of them is depleted; seat the remaining agents altogether. Note that all \textcolor{Red}{Reds} have maximum utility. Suppose no \textcolor{LimeGreen}{Green} has a \textcolor{LimeGreen}{Green} neighbor, i.e.~every \textcolor{LimeGreen}{Green} seats between two \textcolor{blue}{Blues}: no \textcolor{blue}{Blue}, not even a neighbor, would profit from switching with a \textcolor{LimeGreen}{Green}, hence the arrangement is stable. Otherwise, no two \textcolor{blue}{Blues} are neighbors: \textcolor{blue}{Blues} sitting between two \textcolor{Red}{Reds} have maximum utility, and neither \textcolor{blue}{Blues} with a \textcolor{Red}{Red} and a \textcolor{LimeGreen}{Green} neighbor, nor \textcolor{blue}{Blues} with two \textcolor{LimeGreen}{Green} neighbors can improve by exchanging with a \textcolor{LimeGreen}{Green}, even if it is a neighbor. Hence the arrangement is still stable.
    
    \item If $r = b$, alternate between \textcolor{Red}{Reds} and \textcolor{blue}{Blues} until depletion, and then seat all \textcolor{LimeGreen}{Greens} together. All \textcolor{Red}{Reds} have maximum utility, the same holds for all \textcolor{blue}{Blue} except the one with a \textcolor{LimeGreen}{Green} neighbor. This extremal \textcolor{blue}{Blue} cannot improve by switching with a \textcolor{LimeGreen}{Green}, and the arrangement is stable.
\end{itemize}
\end{proof}

We are now left with a single case to handle, without loss of generality the case where \textcolor{Red}{Reds} only like \textcolor{LimeGreen}{Greens}, \textcolor{LimeGreen}{Greens} only like \textcolor{blue}{Blues}, and \textcolor{blue}{Blues} only like \textcolor{Red}{Reds}, and no class likes their own kind. Also without loss of generality, assume $n \geq 4.$ For brevity in what follows, and consistently with the notation of Section \ref{sec:k_classes_algo}, denote \textcolor{Red}{Reds} by 1, \textcolor{LimeGreen}{Greens} by 2, and \textcolor{blue}{Blues} by 3, and write $n_1, n_2$ and $n_3$ for $r, g,$ and $b$, respectively, where $n_1 + n_2 + n_3 = n.$ An arrangement $\pi$ can be seen as a sequence $s = (s_i)_{i \in [n]}$ consisting of $n_1$ ones, $n_2$ twos and $n_3$ threes. Since $s$ represents an arrangement on a cycle, one can assume that $s_0 = s_n$ and $s_{n + 1} = s_1$, etc. We will again lift agent preferences to class preferences. As in Section \ref{sec:binary_pref}, let $W(s) = W(\pi)$ be the utilitarian social welfare induced by $s$/$\pi$. Similarly in spirit to Section \ref{sec:k_classes_algo}, it will be useful to think in terms of long/short-range \emph{blocking triples}. Namely, let $a, b, c, d, e, f$ be agent classes, then we say that triples $(a, b, c)$ and $(d, e, f)$ are \emph{long-range blocking} if $p_b(a) + p_b(c) < p_b(d) + p_b(f)$ and $p_e(d) + p_e(f) < p_e(a) + p_e(c)$; intuitively, neither $b$ wants to swap with $e$, nor vice-versa. Furthermore, for $a, b, c, d$ agent classes, we say that triples $(a, b, c)$ and $(b, c, d)$ are \emph{short-range blocking} if $p_b(a) < p_b(d)$ and $p_c(d) < p_c(a)$; intuitively, if $a, b, c, d$ are consecutive in the arrangement, then neither $b$ wants to swap with $c$, nor vice-versa. An arrangement $s$ is \emph{stable} if the following two conditions hold:
\begin{itemize}
    \item For all $1 \leq i < j \leq n$ such that $1 < j - i < n - 1,$ the triplets $(s_{i - 1}, s_i, s_{i + 1})$ and $(s_{j - 1}, s_j, s_{j + 1})$ are not long-range blocking;
    \item For all $1 \leq i \leq n$ the triplets $(s_{i - 1}, s_i, s_{i + 1})$ and $(s_{i}, s_{i + 1}, s_{i + 2})$ are not short-range blocking.
\end{itemize}

Before proving our main assertion, the following observation will be instrumental:

\begin{proposition}\label{prop:adjacent_seats} Given an arrangement $s,$ the utilitarian social welfare $W(s)$ is equal to $n$ minus the number of pairs of agents from the same class seated in adjacent seats. Formally, $W(s) = n - |\{i \in [n] \mid s_i = s_{i + 1}\}|.$
\end{proposition}
\begin{proof} When two agents belonging to different classes are seated next to each other, by construction of the preferences, exactly one of them will like the other. Moreover, agents do not like agents from the same class as themselves.
\end{proof}

We now prove our assertion, in fact in a stronger form, as follows:

\begin{lemma}\label{lem:3_class_sta:max_util_arg}
Assume 0 only likes 1, 1 only likes 2, and 2 only likes 0, then any arrangement $s$ maximizing the utilitarian social welfare $W(s)$ is stable.
\end{lemma}

\begin{proof} Let $s$ be any arrangement maximizing egalitarian social welfare. Assume for a contradiction that $s$ is not stable. This can be either because of short-range or long-range blocking triples. We tackle the two cases separately:

\begin{enumerate}
    \item Assume $1 \leq i \leq n$ is such that $(s_{i - 1}, s_i, s_{i + 1})$ and $(s_{i}, s_{i + 1}, s_{i + 2})$ are short-range blocking. In other words, write $(s_{i - 1}, s_i, s_{i + 1}, s_{i + 2}) = (a, b, c, d),$ then $p_b(a) < p_b(d)$ and $p_c(d) < p_c(a).$ By symmetry, without loss of generality, assume $b = 1$ and $c = 2.$ From this it follows that $a = 3$ and $d = 2.$ Hence, before $b$ and $c$ swap seats we have $(3, 1, 2, 2)$ while afterward, we have $(3, 2, 1, 2),$ meaning that the swap would increase social welfare (recall Proposition \ref{prop:adjacent_seats}), contradicting the maximality of $W(s).$
    \item Assume $1 \leq i < j \leq n$ are such that $1 < j - i < n - 1$ and $(s_{i - 1}, s_i, s_{i + 1})$ and $(s_{j - 1}, s_j, s_{j + 1})$ are long-range blocking. In other words, write $(s_{i - 1}, s_i, s_{i + 1}) = (a, b, c)$ and $(s_{j - 1}, s_j, s_{j + 1}) = (d, e, f),$ then $p_b(a) + p_b(c) < p_b(d) + p_b(f)$ and $p_e(d) + p_e(f) < p_e(a) + p_e(c).$ By symmetry, without loss of generality, assume $b = 1$ and $e = 2.$ The conditions on the preferences require that there are more threes among $a, c$ than $d, f$ but more twos among $d, f$ than $a, c.$ This implies that at least one of $a$ and $c$ is 3, say $c = 3,$ and at least one of $d$ and $f$ is 2, say $f = 2$. So far, we know that $(a, b, c, d, e, f) = (a, 1, 3, d, 2, 2).$ Looking at the conditions one more time, we get that $(a, 1, 3)$ and $(d, 2, 2)$ are long-range blocking if and only if $d = 3$ implies $a = 3$ and $a = 2$ implies $d = 2.$ We can now list all the possibilities for $(a, b, c, d, e, f)$ and write what would happen after $b$ and $e$ would swap places:  
    \begin{enumerate}
        \item $(1, 1, 3, 1, 2, 2) \rightarrow (1, 2, 3, 1, 1, 2),$ the latter has higher welfare;
        \item $(1, 1, 3, 2, 2, 2) \rightarrow (1, 2, 3, 2, 1, 2),$ the latter has higher welfare;
        \item $(2, 1, 3, 2, 2, 2) \rightarrow (2, 2, 3, 2, 1, 2),$ the latter has higher welfare;
        \item $(3, 1, 3, 1, 2, 2) \rightarrow (3, 2, 3, 1, 1, 2),$ the latter has \textbf{the same} welfare;
        \item $(3, 1, 3, 2, 2, 2) \rightarrow (3, 2, 3, 2, 1, 2),$ the latter has higher welfare;
        \item $(3, 1, 3, 3, 2, 2) \rightarrow (3, 2, 3, 3, 1, 2),$ the latter has higher welfare.
    \end{enumerate}
    Hence, in all cases except $(a, b, c, d, e, f) = (3, 1, 3, 1, 2, 2)$, we get a contradiction of the maximality of $W(s).$ We now need to handle this last case, which indeed requires further insight. For this case, we will now show that we can construct an arrangement $s'$ such that $W(s') > W(s),$ again contradicting the maximality of $W(s).$ The technique we will use, however, is different from the usual swapping of the seats of $b$ and $e.$ Namely, let $x$ be the agent sitting to the left of $3, 1, 3.$ For illustration, if we were to write $s$ in full, it would be $\ldots, x, (3, 1), 3, \ldots, 1, 2, 2, \ldots.$ We now take the two bracketed agents $(3, 1)$ and reseat them between agents $2, 2.$ The seating arrangement $s'$ obtained as a result looks as follows:~$\ldots, x, 3, \ldots, 1, 2, (3, 1), 2, \ldots.$ Comparing the social welfares $W(s)$ and $W(s')$ notice that the moved pair $(3, 1)$ has ``broken'' the two adjacent twos $2, 2.$ Moreover, irrespective of the value of $x,$ in $s$ agent $x$ is sitting next to a 3, and the same is true in $s'.$ As a result, by Proposition \ref{prop:adjacent_seats}, we get that $W(s') - W(s) = 1,$ contradicting the maximality of $W(s).$
\end{enumerate}

\end{proof}

\subsubsection{Proving Theorem \ref{thm:3_class_2_val_unst_path}.} We now complete the proof of Theorem \ref{thm:3_class_2_val_unst_path}, restated below for convenience.

\repeattheorem{thm:3_class_2_val_unst_path}
\begin{proof}[continued]
Suppose Alice sits next to Bob. On one hand, if Alice is seated at an endpoint of the path, then both Alice and Bob would improve by exchanging seats: after the exchange, Alice gets utility zero instead of minus one, whereas Bob gets utility one instead of zero. On the other hand, if Alice is not extremal, there is a Friend at one endpoint of the path not a neighbor of Bob: swapping seats with him would increase Alice's utility from zero to one, while the Friend would improve from one (respectively minus one if its only neighbor is Alice) to two (respectively zero).

Now suppose Bob is sitting away from Alice:~exchanging with a neighbor of Alice would increase its utility from at most minus one to at least zero, whereas the Friend would improve from at most zero to at least one. Hence no arrangement on a path is stable. 
\end{proof} 

\subsubsection{Proving Lemma \ref{lem:4_non_cv}.} We now prove Lemma \ref{lem:4_non_cv}, restated below for convenience.

\repeatlemma{lem:4_non_cv}
\begin{proof}
Let $\pi^*$ be the arrangement $(a,~b,~c,~d).$ Since each agent approves of exactly one other agent, they can never get utility strictly greater than one. Since only $d$ does not achieve utility one in $\pi^*$, it is stable. Moreover, consider arrangements $\pi_1 = (a,~d,~b,~c)$ and $\pi_2 = (c,~d,~b,~ a).$ The only blocking pair in $\pi_1$ is $(a,~c)$ as both $b$ and $d$ have utility one. Exchanging them leads to arrangement $\pi_2.$ Similarly, in $\pi_2$, the only blocking pair is $(b,~d).$ Exchanging them gives $\pi_1$ back, up to reversal of the seat numbers. Hence, the swap dynamics cannot converge. See Figure \ref{fig:4_non_cv_exple} for an illustration.
\end{proof}

\section{Stability Analysis of Profile $\mathcal{P}_5$} \label{sec:p5_stability_analysis}

\begin{figure}[t]
    \centering
    \includegraphics[width=0.85\linewidth]{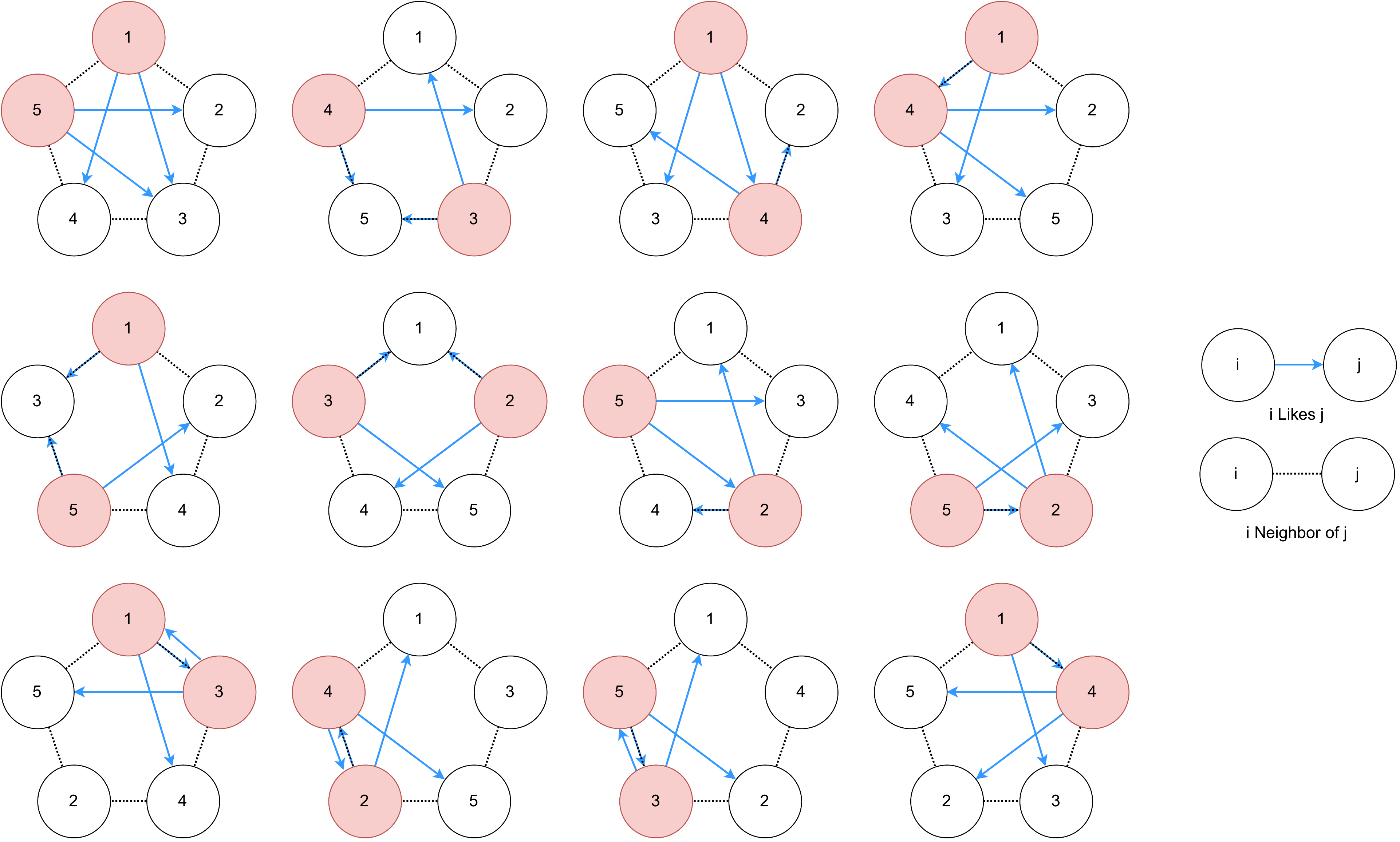}
    \caption{Pictorial proof of the instability of $\mathcal{P}_5.$}
    \label{fig:proof_5_unst}
\end{figure}

In this appendix, we briefly analyze the stability of profile $\mathcal{P}_5$ from Figure \ref{fig:5_unst}, which is the only preference profile with $n = 5$ that is unstable on a cycle. Despite our best efforts, there does not seem to be an easy explanation for the emergence of instability in this case.

\begin{lemma}\label{lem:5_unst} Binary preferences $\mathcal{P}_5$  induce no stable arrangements on a cycle.
\end{lemma}
\begin{proof}
Figure \ref{fig:proof_5_unst} provides a pictorial proof of the instability of $\mathcal{P}_5.$ Each of the twelve possible cyclic arrangements is displayed and a blocking pair is shown in red. The outgoing edges of the nodes participating in the blocking pairs are shown in blue for effortless verification.
\end{proof}

\section{Evidence for Exponential Convergence} \label{sec:exponential_convergence}

In this appendix, we revisit the potential function used to prove convergence of the swap dynamics in the proof of Theorem \ref{thm:close-swaps-two-valued-is-stable}. In essence, we show that each increase in potential might be very small, leading to an exponential number of increases. Note, however, that we do not exhibit preference instances where such exponential behavior can be observed.\footnote{We did not manage to construct such preferences.} Moreover, we also do not exclude the possibility that the function can be shown to increase enough with each swap on average when the dynamics are carried out. What we show, instead, is that there is an exponentially long chain of potential values such that between any two consecutive values in the chain there are preferences for which the transition could occur by exchanging a blocking pair.

\begin{lemma}
The potential argument alone cannot guarantee polynomial time convergence of swap dynamics.
\end{lemma}
\begin{proof}
We first consider the auxiliary potential sequence $S(\pi)$ with values in $\{1,2\},$ and study the effect of exchanges at distances one and two that keep the social welfare constant. To simplify notation, we map $S(\pi) \in \{1,2\}^{n-1}$ to $S_b(\pi) \in \{0,1\}^{n-1}$ by subtracting one.

Performing an exchange at distance one while keeping the social welfare constant corresponds to the following modification of a subsequence of $S_b(\pi)$: $0x1 \to 1\overline{x}0$, where $x\in \{0,1\}$ (see Figure \ref{fig:exchange_dist_1_potential}). Indeed, an exchange at distance one modifies the utility of at most four people, hence modifies a subsequence of length at most three of the auxiliary potential, the latter being defined not on vertices but edges; we call this operation $f_3.$
$$f_3: 0x1 \to 1\overline{x}0$$ 

Similarly, we define the operation corresponding to exchanging at distance two that keeps the social welfare constant:~it maps the subsequence of $S_b(\pi)$ of length four $0xy1$ to $1\overline{y}\overline{x}0$. For the same reason as above, we call this operator $f_4.$
$$f_4:0xy1 \to 1\overline{y}\overline{x}0$$

In the following, we denote by ``Apply $f_i$ at position $j$'' the application of $f_i$ on the subsequence of indices $[j,j+i].$

Based on those two operators, we further define the operator $f_8$ mapping the sequence of length 8 $(0,0,0,0,0,0,0,1)$ to $(1,0,0,0,0,0,0,0)$. It consists of the following operations: 
\begin{enumerate}
    \item Apply $f_3$ at position $6$: $(0,0,0,0,0,0,0,1) \to (0,0,0,0,0,1,1,0);$
    \item Apply $f_4$ at position $3$: $(0,0,0,0,0,1,1,0) \to (0,0,1,1,1,0,1,0);$
    \item Apply $f_4$ at position $2$: $(0,0,1,1,1,0,1,0) \to (0,1,0,0,0,0,1,0);$
    \item Apply $f_4$ at position $4$: $(0,1,0,0,0,0,1,0) \to (0,1,0,1,1,1,0,0);$
    \item Apply $f_4$ at position $3$: $(0,1,0,1,1,1,0,0) \to (0,1,1,0,0,0,0,0);$
    \item Apply $f_3$ at position $1$: $(0,1,1,0,0,0,0,0) \to (1,0,0,0,0,0,0,0).$
\end{enumerate}

For $k\geq 3$, we then recursively define the operators $f_{3k-1}$ mapping the sequence of length $3k-1$ $(0,\dots,0,1)$ to $(1,0,\dots,0)$ through the following operations:
\begin{enumerate}
    \item Apply $f_3$ at position $3(k-1)$: $(0,\dots,0,1)\to (0,\dots,0,1,1,0);$
    \item Apply $f_{3(k-1)-1}$ at position $2$: $(0,\dots,0,1,1,0)\to (0,1,0\dots,0,1,0);$
    \item Apply $f_{3(k-1)-1}$ at position $3$: $(0,1,0\dots,0,1,0) \to (0,1,1,0,\dots,0);$
    \item Apply $f_3$ at position $1$: $(0,1,1,0,\dots,0)\to (1,0,\dots,0).$
\end{enumerate}

Note that this is well-defined since $f_8$ is defined separately.

Noting that operator $f_{3k-1}$ contains more than $2^{k-1}$ uses of $f_3$ and $f_4$ for all $k\geq 3$ concludes the proof.
\end{proof}

\begin{figure}[t]
    \centering
    \includegraphics[width=0.6\linewidth]{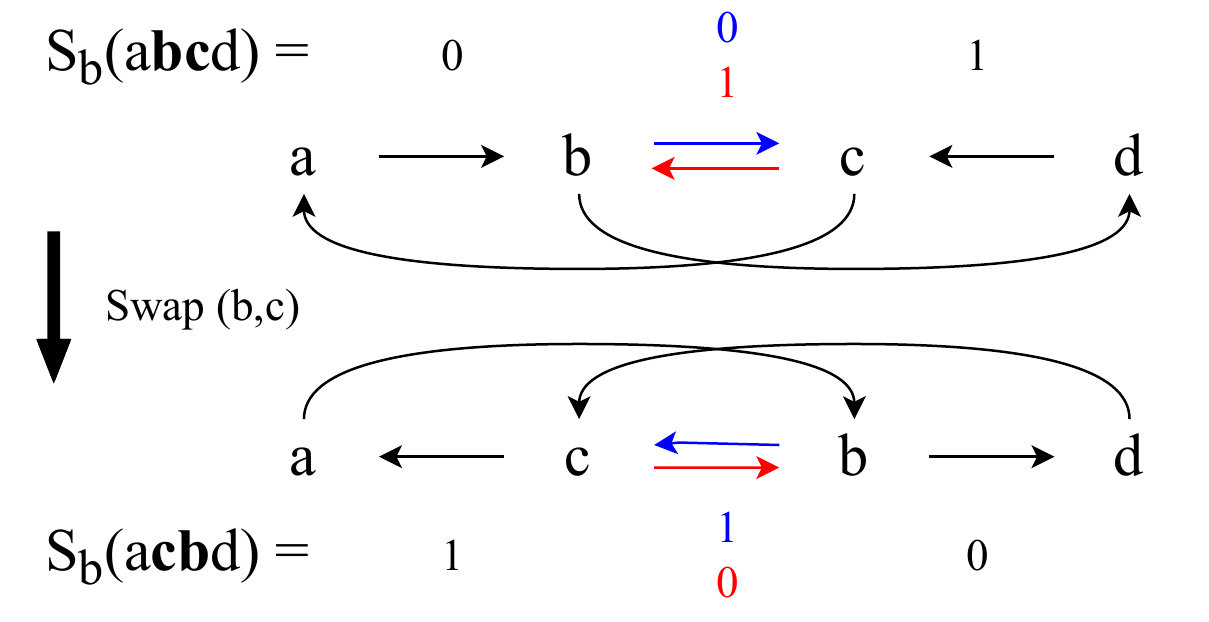}
    \caption{Change of $S_b(\pi)$ after an exchange at distance one preserving the social welfare.}
    \label{fig:exchange_dist_1_potential}
\end{figure}

\section{Non-Monotonicity of Stability}\label{sec:non_mono_sta}
In this section, we show that stability is non-monotonic, 
i.e., adding agents to a given instance can both introduce or destroy stability. For instance, consider preferences $\mathcal{P}_5$ in Figure \ref{fig:5_unst}. On a cycle, all binary preferences with either four or six agents possess a stable arrangement, implying that adding a fifth agent could destroy stability while adding a sixth agent would restore it. We now show that this phenomenon can occur for all values $n\geq 7$:

\begin{theorem}\label{thm:non_monotonic_sta}
    For all $n\geq 7$, adding an agent can both destroy stability or restore it.
\end{theorem}
\begin{proof}
We first prove that for $n\geq 7$ adding an agent can break stability. Consider the unstable preferences of size $n+1$ in Theorem \ref{thm:bin_unst_cycle}:~we show that removing one agent creates a stable arrangement. Indeed, remove agent $a$ and consider sitting $c$ between $b_1$ and $b_2$:~everyone but $c$ has maximum utility (since $b_1$ and $b_2$ love only one person, they cannot get utilities higher than 1), hence $c$ cannot exchange seats with anyone and the arrangement is stable.

We now prove that for $n\geq 7$ adding an agent can create stability. Consider the unstable preferences of size $n$ in Theorem \ref{thm:bin_unst_cycle}, and add one extra agent $b_3 \in B;$ i.e., which only likes $a$ and $c$ and is only liked by the members of $D.$ Consider placing $b_1, ~c, ~b_2, ~a, ~b_3$ consecutively on the cycle in this order:~agent $b_2$ and all members of $D$ have utility 2, hence only $b_1$, $b_3$ or $c$ can be part of a blocking pair. However, $b_1$ and $b_3$ both have utility one and would therefore only exchange for a seat with utility 2:~there exists only one such seat, currently held by $b_2,$ who has maximum utility and would hence not agree to swap places. As a result, neither $b_1$ nor $b_3$ are part of a blocking pair, so the arrangement is stable.
\end{proof}

\section{Blockwise-Diagonal Preferences}\label{app:blockwise}
In this section, we show that, even when the preferences of the agents are, in a certain sense, highly decomposable, knowledge about the stability of the subparts is unlikely to help us to find stable arrangements for the instance as a whole. In essence, we show why non-monotonicity can make reasoning about stability rather challenging, even in reasonably simple cases. To make the previous statements precise, we introduce the concept of agent \emph{components}, as follows.

\begin{definition}
We say that a set of agents $\mathcal{C} \subseteq \mathcal{A}$ is isolated if for all $a \in \mathcal{C}$ and $b \in \mathcal{A}\setminus \mathcal{C},$ it holds that $p_{a}(b) = p_{a}(b) = 0.$ Set $\mathcal{C}$ is called a component if none of its proper subsets are isolated.
\end{definition}

Note also that components and classes are two different notions:~agents in the same component may not have the exact same preferences, but instead are limited to only caring about agents in their component. Assume that the set of agents $\mathcal{A}$ is partitioned into components $\mathcal{A} = \mathcal{C}_1 \cup \ldots \cup \mathcal{C}_k$ (this partition always exists and is unique).\footnote{This is because the components correspond to the connected components of the undirected graph with vertex set $\mathcal{A}$ and edges $(a, b)$ for any two distinct agents such that either $u_a(b) \neq 0$ or $u_b(a) \neq 0.$} The preference matrix can then be represented, after a potential reordering of the agents, as a blockwise-diagonal matrix. When the partition into components is non-trivial; i.e., $k > 1$; intuitively, finding an arrangement that is stable for such preferences should be easier than for general ones:~first, find a stable arrangement on a path for each component, and then join all those paths to obtain a stable arrangement on a cycle (or on a path). While this method is indeed guaranteed to produce a stable arrangement whenever each component admits a path stable arrangement (at least for non-negative preference values), we will actually show that there are many instances where a stable arrangement exists but can not be produced by this approach. Before showing this, we need a technical lemma for cycles, stated next.

\begin{lemma} \label{lem:0_util_stability} Let $\pi$ be an arrangement where each agent sits between two agents from different components.
If for any two distinct components $\mathcal{C}_i$ and $\mathcal{C}_j$ there is at most one pair of agents $(a, b) \in \mathcal{C}_i \times \mathcal{C}_j$ such that $a$ and $b$ are neighbors in $\pi$, then $\pi$ is stable on a cycle.
\end{lemma}

\begin{proof} First, notice that in such an arrangement all agents get utility zero. Let $a \neq b$ be two agents. We want to show that they do not form a blocking pair. Assume $i, j \in [k]$ are such that $a \in \mathcal{C}_i$ and $b \in \mathcal{C}_j.$ If $i = j$, then by assumption $a$ would also have utility zero when seating in $b$'s seat, so $(a, b)$ is not a blocking pair. Now, assume $i\neq j.$ Let $a_{\ell}$ and $a_{r}$ denote the two neighbors of $a$ on $\pi$; we define $b_{\ell}$ and $b_{r}$ similarly. Suppose $a$ and $b$ wanted to switch place:~this means $a$ would have strictly positive utility at $b$'s seat, and $b$ would have strictly positive utility at $a$'s seat. Assuming $a$ and $b$ are not neighbors, this means that $\{a_\ell, a_r\} \cap \mathcal{C}_j \neq \varnothing$ and  $\{b_\ell, b_r\} \cap \mathcal{C}_i \neq \varnothing.$ Let $a'$ and $b'$ be such that $a' \in \{a_\ell, a_r\} \cap \mathcal{C}_j$ and $b' \in \{b_\ell, b_r\} \cap \mathcal{C}_i.$ Since $(a, a')$ and $(b', b)$ are both in $\mathcal{C}_i \times \mathcal{C}_j$ and are distinct pairs of adjacent agents in $\pi,$ this contradicts the hypothesis. Furthermore, if $a$ and $b$ are neighbors, then $b$ envying $a$ implies that $a$'s second neighbor is also in  $\mathcal{C}_j$, from which the pairs of agents formed by $a$ and its neighbors similarly contradict the hypothesis. Therefore, the arrangement is stable.
\end{proof}

Armed as such, we now show that any preference profile inducing no stable arrangements on a path can be used to construct preferences whose components all resemble this profile and yet the joint profile admits an arrangement that is stable on a cycle. In other words, even when a profile decomposes non-trivially into components, knowledge about the stability of the components does not necessarily help in resolving the stability of the instance as a whole.

\begin{theorem}\label{thm:euler_tour_components} Let $\mathcal{P}_{\mathit{path}} \in \{0,1\}^{\ell \times \ell}$ be a one-component preference profile such that no stable arrangement on a path exists. Then, there exists a preference profile $\mathcal{P}$ whose components all resemble $\mathcal{P}_{\mathit{path}}$ admitting a stable arrangement on a cycle.
\end{theorem}

\begin{proof} We construct a larger blockwise-diagonal preference matrix $\mathcal{P}\in \{0,1\}^{k\times k }$ by copying $\mathcal{P}_{\mathit{path}}$ a number $k =2\ell +1$ of times:
$$\mathcal{P} =\left[\begin{array}{ c  c  c  c}
     \mathcal{P}_{\mathit{path}}^{(1)}   &           &       & 0           \\
                & \mathcal{P}_{\mathit{path}}^{(2)}  &       &             \\
                &           & \ddots   &             \\
    0           &           &       & \mathcal{P}_{\mathit{path}}^{(k)}
  \end{array}\right]$$
  
Note that $\mathcal{P}_{\mathit{path}}^{(1)}, \ldots, \mathcal{P}_{\mathit{path}}^{(k)}$ naturally gives a partition into components $\mathcal{C}_1,\ldots ,\mathcal{C}_k$; moreover, the choice of $\mathcal{P}_{\mathit{path}}$ immediately gives that all $\mathcal{C}_1, \ldots, \mathcal{C}_k$ are unstable on paths.

We now construct a cycle stable arrangement for $\mathcal{P}.$ Since $k$ is odd, graph $K_k$, which is the undirected clique graph with $k$ vertices, possesses an Euler tour $T$ where each vertex is visited $\ell$ times. We construct the arrangement $\pi$ (more precisely $\pi^{-1}$) on the cycle by replacing every occurrence of $i$ in $T$ by an agent $a_i \in \mathcal{C}_i$, without repetition, starting for example from seat one. For the $\pi$ we have just constructed, it then holds that the conditions to apply Lemma \ref{lem:0_util_stability} are satisfied since the Euler Tour traverses each edge in $K_k$ precisely once. By the lemma, $\pi$ is a cycle stable arrangement of $\mathcal{P}.$
\end{proof}

\section{Stability of Random Binary Preferences}\label{app:proba_method}
In this chapter, we conduct a study of stability using probabilistic tools. In particular, we employ the Lovász Local Lemma:

\begin{lemma}[Lovász Local Lemma]\label{lem:LLL}
Let $A_1, A_2,\ldots, A_k$ be a sequence of events such that each event occurs with probability at most $p$ and  is independent of all but at most $d$ other events. If $epd <1$, then the probability that none of the events occurs $P \left(\cap_{i=1}^k \overline{A_i}\right) $ is greater than or equal to $ \left(1-\frac{1}{d+1}\right)^k .$
\end{lemma}

We now give a lower bound on the expected number of arrangements stable on a cycle when the preference graph is sampled from the Erd\H{o}s-Rényi model $G(n, p)$ with average node degree either $O(\sqrt{n})$ or $n - O(\sqrt{n}).$

\begin{theorem}
    Suppose a binary preference graph $\mathcal{P}$ is drawn at random using the Erd\H{o}s-Rényi model $G(n,p)$  with  $p \leq Cn^{-1/2},$ where $C=(96e)^{-1/2}.$
    
    \noindent Then, the expected number of stable arrangements on a cycle is at least: 
    
    $$\frac{1}{2}(n-1)! \exp \left( -\frac{n(n-1)}{2n-3} \right).$$
    
    \noindent The same results holds for $p \geq 1-Cn^{-1/2}.$
\end{theorem}
\begin{proof}
    Let $S=\sum_\pi S_\pi$ be the random variable counting the stable arrangements on a cycle for preference $\mathcal{P}$, where $S_\pi = 1$ if arrangement $\pi$ is stable and $0$ otherwise. Since all permutations of $\mathcal{P}$ follow the same distribution, all $S_\pi$ have the same expectation and we only need to consider the identity arrangement $\pi = \mathit{id}$.

    \noindent For $i,j \in [n]$, event $L_{ij}$ corresponds to the $i^\textsuperscript{th}$ agent liking the  $j^\textsuperscript{th}$ one, event $E_{ij}$ to the $i^\textsuperscript{th}$ agent envying the  $j^\textsuperscript{th}$ one, and event $B_{ij}$ to agents $(i,j)$ forming a blocking pair.
   Finally, let the random variable $U_i$ denote the utility of the $i^\textsuperscript{th}$ agent. First, note that events $E_{ij}$ and $E_{kl}$ are independent for all $i\neq k$; from which $Pr(E_{ij})=Pr(E_{ji}).$ 
    We therefore get:
    \begin{equation}\label{equ:Bij_E_ij2}
    \begin{aligned}
    Pr(B_{ij}) &= Pr(E_{ij}\cap E_{ji}) \\ 
    &= Pr(E_{ij})Pr(E_{ji})\\ 
    &= Pr(E_{ij})^2
    \end{aligned}
    \end{equation}
    By symmetry, we only have to calculate $(Pr(E_{1j}))_{2\leq j \leq n}.$\\
    If $j \in \{2,3\}$: $$Pr(E_{1j}) = Pr\left(E_{1j} \mid \overline{L_{1n}}\right)Pr\left(\overline{L_{1n}}\right) = p(1-p)$$
    If $j \in \{n-1,n\}$: $$Pr(E_{1j}) = Pr\left(E_{1j} \mid \overline{L_{12}}\right)Pr\left(\overline{L_{12}}\right) = p(1-p)$$
    If $3< j <n-1$:
    \begin{equation*}
    \begin{aligned}
    Pr(E_{1j}) &= Pr\left(E_{1j} \mid U_1 = 0\right) Pr(U_1=0) + Pr\left(E_{1j} \mid U_1 = 1\right) Pr(U_1=1) \\
    &= (2p(1-p)+p^2)(1-p)^2 + p^2 2p(1-p)
    \end{aligned}
    \end{equation*}

    \noindent Together with Equation \eqref{equ:Bij_E_ij2}, we subsequently get that:
    \begin{equation}\label{equ:Bij}
    Pr(B_{ij}) = \left\{
    \begin{array}{ll}
        p^2(1-p)^2 & \mbox{if } |i-j| \leq 2\\
        \left(2p^3(1-p) + 2p(1-p)^3 + p^2(1-p)^2 \right)^2 & \mbox{otherwise.}
    \end{array}
    \right.
    \end{equation}

    \noindent Now, consider the family of events $(B_{ij})_{1\leq i<j\leq n}. $ Note that each even $B_{ij}$ is independent of events $B_{kl}$ where $\{i,j\}\cap\{k,l\} =\varnothing$, but dependent of events with which it shares an index, so $B_{ij}$ depends on at most $d=2(n-2)$ other events. Therefore, when $Pr(B_{ij}) < \frac{1}{2e(n-2)}$, Lemma \ref{lem:LLL} gives:
    \begin{equation}
    \begin{aligned}
    \mathbb{E}[S_\pi] = Pr(S_\pi=1) & \geq \left(1-\frac{1}{2n-3}\right)^{\frac{n(n-1)}{2}} \\
    &\geq \exp \left( -\frac{n(n-1)}{2n-3} \right)\\
    \end{aligned}
    \end{equation}
By linearity of expectation, we finally get:

\begin{equation}
   \mathbb{E}[S] \geq \frac{1}{2}(n-1)! \exp \left( -\frac{n(n-1)}{2n-3} \right)
\end{equation}
    
    \noindent It is only left to verify that $Pr(B_{ij}) \leq \frac{1}{2e(n-2)}.$ Note that $B_{ij}$ has at most six terms all strictly smaller than $8p^2.$ Hence, for $p \leq \dfrac{1}{\sqrt{96en}} \leq \dfrac{1}{\sqrt{96e(n-2)}},$ we have:
    \begin{equation*}
    \begin{aligned}
    Pr(B_{ij}) & < 6\times 8 p^2 \\
    & \leq \frac{48}{96e(n-2)} = \frac{1}{2e(n-2)}
    \end{aligned}
    \end{equation*}
    Comparison  with $6\times8(1-p)^2$ instead gives the result for $p \geq 1-\dfrac{1}{\sqrt{96en}}.$
\end{proof}

In practice, our result implies that, for random preferences of average out-degree at most $O(\sqrt{n}),$ a naive approach sampling arrangements uniformly at random on average determines a stable arrangement using exponentially fewer samples than the theoretically required $(n - 1)! / 2$.

\section{Z3 Solver for Non-Negative Two-Valued Preferences}\label{app:Z3_solver}

In this section, we describe the Z3 solver\footnote{\url{https://github.com/Z3Prover/z3}} employed in Section \ref{sec:binary_pref} to check whether all non-negative two-valued preferences for $n \leq 7$ are stable on paths and cycles. We only describe the case of cycles, as for paths it is enough to add one additional agent with null preferences from and toward all other agents. We note that for cycles it is enough to consider the binary case. For paths it is enough to consider the cases $\Gamma \in \{\{0, 1\}, \{1, 2\}, \{1, 3\}, \{2, 3\}\}.$ For clarity, we describe here only the binary case, but the necessary modifications for the other cases are straightforward.

Listing \ref{lst:Z3_solver} shows the main body of the solver. In line 8, we introduce a function associating the boolean ``$i$ likes $j$'' to each pair $(i, j).$ In line 14, we define an array of $n$ integers encoding the index of the agent placed in each seat; lines 17 to 24 constrain this array to be a permutation representing one of the $(n-1)!/2$ cycles. 
In lines 15, 16, and 26 we implement the main constraint:~all arrangements on a cycle must induce a blocking pair.

Note that all simultaneous permutations of rows and columns of a solution are themselves solutions, as this corresponds to relabeling the agents. Therefore, it is desirable for efficiency to implement some kind of symmetry breaking. We do this in 
Listing \ref{lst:SymBreak} by requiring that the agents are sorted by the number of agents they approve of, breaking ties by the number of agents that approve them.

Finally, Listing \ref{lst:isBlocking} shows how to check whether agent-pair $(i, j)$ is a blocking. Note that adjacent and non-adjacent seats require different treatments, and so lead to different logical expressions.

Finally, note that Z3 returns either ``Unsatisfiable'' when no solutions exist, or ``Satisfiable'' and one solution otherwise. Finding all possible solutions is therefore rather tedious:~after finding a solution, to get another one, we need to add a constraint that ``the preferences are \textbf{not} these ones.'' This detail is omitted for brevity.

\usemintedstyle{xcode}
\setminted{fontsize=\footnotesize, linenos=TRUE}
\begin{listing}[h]
\begin{minted}{Python}
from z3 import *

N = 7

def UnstableCounterExample():
    s = Solver()
    # Pref(i, j): bool("Agent i likes agent j").
    Pref = Function('Pref', IntSort(), IntSort(), BoolSort())
    # Constrains Pref(i, i) = False.
    s.add([Not(Pref(i, i)) for i in range(N)])
    # Symmetry breaking on Pref.
    s.add(SymBreakSumRowCol(Pref))
    # Seat i is taken by agent Arr[i].
    Arr = [Int(f"Arr_{i}") for i in range(N)]
    s.add(ForAll(
        Arr, Implies(
            # Seat 0 is taken by agent 0.
            And(Arr[0] == 0,
                # Arr is a permutation.
                And([Arr[i] > 0 for i in range(1, N)]),
                And([Arr[i] < N for i in range(1, N)]),
                Distinct([Arr[i] for i in range(N)]),
                # Symmetry breaking on Arr.
                Arr[1] < Arr[N - 1]),
            # Agents at seats (i, j) form a blocking pair.
            Or([isBlockingPair(Pref, Arr, i, j) for i in range(1, N) for j in range(i)]))))
    SolveAndPrint(s)

def SolveAndPrint(s):
    print("Solving", s)
    val = s.check()
    if val == sat:
        print("Satisfiable", s.model())
    elif val == unsat:
        print("Unsatisfiable")
    else:
        print("Unknown")

UnstableCounterExample()
\end{minted}
   \caption{Main body of the solver.}
   \label{lst:Z3_solver}
\end{listing}

\begin{listing}
\begin{minted}{Python}
def SymBreakSumRowCol(pref):
    """Returns True iff rows of pref have increasing sums, and in case of 
    equality, the sums on the corresponding columns are increasing."""
    sums_row = [Sum([pref(i, j) for j in range(N)]) for i in range(N)]
    sums_col = [Sum([pref(i, j) for i in range(N)]) for j in range(N)]
    # First order on the sums of the rows.
    sum_row_ord = And([sums_row[i] <= sums_row[i + 1] for i in range(N - 1)])
    # Then order on the sums of the columns.
    sum_col_ord = And([Implies(sums_row[i] == sums_row[i + 1], sums_col[i] <= sums_col[i])
                       for i in range(N - 1)])
    return And(sum_row_ord, sum_col_ord)
\end{minted}
   \caption{Symmetry breaking of preferences.}
   \label{lst:SymBreak}
\end{listing}

\begin{listing}
\begin{minted}{Python}
def isBlockingPair(pref, arr, i, j):
    """Returns True if agents at seats i and j form a blocking pair.
    We suppose 0 <= j < i <= N - 1."""
    if j == i - 1:                  # Seats i and j adjacent.
        blockpair = And(Not(pref(arr[i], arr[(i + 1) % N])), pref(arr[i], arr[(j - 1 + N) % N]),
                                    # Agent at seat i envies agent at seat j.
                        Not(pref(arr[j], arr[(j - 1 + N) % N])), pref(arr[j], arr[(i + 1) % N]))
                                    # Agent at seat j envies agent at seat i.
        return blockpair
    if j == 0 and i == N - 1:       # Seats i and j adjacent.
        blockpair = And(Not(pref(arr[i], arr[(i - 1 + N) % N])), pref(arr[i], arr[(j + 1) % N]),
                                    # Agent at seat i envies agent at seat j.
                        Not(pref(arr[j], arr[(j + 1) % N])), pref(arr[j], arr[(i - 1 + N) % N]))
                                    # Agent at seat j envies agent at seat i.
        return blockpair
    else:                           # Seats i and j non adjacent.
        blockpair = And(Sum(pref(arr[i], arr[(i - 1 + N) % N]), pref(arr[i], arr[(i + 1) % N]))
                        < Sum(pref(arr[i], arr[(j - 1 + N) % N]), pref(arr[i], arr[(j + 1) % N])),
                                    # Agent at seat i envies agent at seat j.
                        Sum(pref(arr[j], arr[(j - 1 + N) % N]), pref(arr[j], arr[(j + 1) % N]))
                        < Sum(pref(arr[j], arr[(i - 1 + N) % N]), pref(arr[j], arr[(i + 1) % N])))
                                    # Agent at seat j envies agent at seat i.
        return blockpair
\end{minted}
   \caption{
   Testing for a blocking pair.
   }
   \label{lst:isBlocking}
\end{listing}

\end{document}